\documentclass{amsart}

\title
{The Riemann theta function near soliton limit}

\author{Yuji Kodama} 
\date{\today}
\address{Department of Mathematics, The Ohio State University, Columbus, OH 43210, USA}
\email{kodama@math.ohio-state.edu}
\subjclass[2000]{}

\usepackage{url}
\usepackage{color}
\usepackage{rotating}
\newcommand{\rotxc}[1]{\begin{sideways}#1\end{sideways}}
\newcommand{\invert}[1]{\rotxc{\rotxc{#1}}}

\usepackage{float}
\usepackage{subfigure}

\def\Le{\hbox{\invert{$\Gamma$}}}
\countdef\x=23
\countdef\y=24
\countdef\z=25
\countdef\t=26

\def\tbox(#1,#2)#3{
\x=#1 \y=#2 
\multiply\x by 12 
\multiply\y by 12 
\z=\x \t=\y
\advance\z by 12 
\advance\t by 12 
\psline(\x,\y)(\x,\t)(\z,\t)(\z,\y)(\x,\y)
\advance\x by 6
\advance\y by 6 
\rput(\x,\y){{\bf #3}}}

\usepackage{amssymb}
\usepackage{graphicx}
\usepackage{epstopdf}
\usepackage{amscd}
\usepackage{appendix}
\usepackage{graphics}
\usepackage{tikz}
\usepackage{caption}
\def\proof{\par{\it Proof}. \ignorespaces}
\def\endproof{{\ \vbox{\hrule\hbox{%
     \vrule height1.3ex\hskip0.8ex\vrule}\hrule }}\par}

\theoremstyle{definition}

\theoremstyle{remark}

\numberwithin{equation}{section}

\let\trueint=\int
\let\truesum=\sum
\def\int{\mathop{\textstyle\trueint}\limits}
\def\sum{\mathop{\textstyle\truesum}\limits}
\let\<=\langle
\let\>=\rangle

\def\P{{\mathbb P}}

\def\RR{{\mathcal R}}

\def\sech{\mathop{\rm sech}\nolimits}

\def\t{\mathbf{t}}
\def\0{\mathbf{0}}
\def\z{\mathbf{z}}

\def\wb{\hskip0.3cm\circle{8}}

\def\wbs{\hskip0.15cm\circle{4}}

\def\edge{\ar@{-}}
\def\dedge{\ar@{.}}

\usepackage{amsmath, amsthm, amssymb, amsbsy,amsfonts,amscd}
\usepackage[mathscr]{eucal}
\usepackage{epsfig}
\usepackage{young}
\newtheorem{theorem}{Theorem}[section]

\newtheorem{proposition}[theorem]{Proposition}
\newtheorem{lemma}[theorem]{Lemma}

\newtheorem{remark}[theorem]{Remark}

\usepackage{amsfonts}
\usepackage{xy}
\usepackage{amssymb}
\usepackage{amsmath}

\newcommand{\Z}{\mathbb Z}
\newcommand{\R}{\mathbb R}

\newcommand{\C}{\mathbb C}

\newcommand{\thmrefer}[1]{\renewcommand\thetheorem
  {\protect\ref{#1}}\addtocounter{theorem}{-1}}

\xyoption{all}
\CompileMatrices

\begin{document}

\begin{abstract}
It has been known that (the square of) Jacobi elliptic function can be expressed by an infinite sum of single solitons of $\sech^2$ shape. In this paper, we show that there exists a similar structure for higher genus cases. It turns out that this is just a consequence of the quasi-periodicity of the Riemann $\theta$-function.
More precisely, we show that the second derivative of $\log \theta$ with the \emph{real} Riemann $\theta$-function of genus $g$ near soliton limit can be well approximated by the \emph{sum} of \emph{real} and 
\emph{regular} $g$-soliton solutions of Hirota-type in $g$-dimensional real space $\R^g$. This leads to
a tessellation of $\R^g$, whose tile is an oblique prism divided into $2^g$ sections by
 hyper-planes of dominant exponents in the theta function. We apply the results to study quasi-periodic solutions to the KdV and KP equations. We construct the quasi-periodic solutions using the Schottky group, which uniformizes the corresponding Riemann surfaces.
We also discuss solitons on quasi-periodic background by pinching some of the homological cycles.
\end{abstract}

\maketitle

\tableofcontents


\section{Introduction}

We start with an overview of the present paper using the simplest case of genus 1 
(denoted as $g=1$) periodic solution expressed by the Jacobi theta function.
It is well-known and surprising that the Jacobi elliptic function, e.g.  $\text{cn}^2$ can be expressed by the sum
of $\text{sech}^2$ functions \cite{T:75}. In the integrable system community, this is often referred to as "\emph{a cnoidal wave can be expressed by an infinite sum of one solitons}. We give a brief review on this fact for the sake of our present paper. Let $\theta(z;\Omega)$ be the Jacobi theta function
on $z\in \C$,
\begin{equation}\label{eq:Jacobi}
\theta(z;\Omega)=\sum_{m=-\infty}^{\infty}\exp\,2\pi i\left(\frac{1}{2}m^2\Omega+mz\right),
\end{equation}
where $\Omega\in \C$ is the moduli parameter satisfying $\text{Im}(\Omega)>0$. 
In this paper, we mainly consider the case with large $|\Omega|$.
Then as mention in \cite{Mu:84} (see also \cite{K:24}),
we consider the (half) shift $z\to z-\frac{\Omega}{2}$ to get a nontrivial limit as $\Omega\to i\,\infty$, that is,
\begin{align}\label{eq:ShiftedT}
\theta\left(z-\frac{\Omega}{2};\Omega\right)&=\sum_{m=-\infty}^{\infty}\exp\,2\pi i\left(\frac{1}{2}m(m-1)\Omega+mz\right)\\
&\rightarrow ~ 1+e^{2\pi i z}\nonumber
\end{align}
Note that \eqref{eq:Jacobi} is normally referred to as the Jacobi $\theta_3(z;\Omega)$ and 
the shifted $\theta_3$-function is related to the Jacobi $\theta_2(z;\Omega)$ as (see e.g. \cite{Mu:83})
\[
\theta_3\left(z-\frac{\Omega}{2};\Omega\right)=e^{-\frac{1}{4}\pi i\Omega +\pi i z}\theta_2(z;\Omega).
\]
In this paper, we consider the theta function \eqref{eq:ShiftedT} on $i\R$ ($z\in i\R$) with the real variables $\varphi$ and $T$ defined by
\begin{equation}\label{eq:phiT}
\varphi:=2\pi i z \qquad\text{and}\qquad  T:=-2\pi i\Omega >0.
\end{equation}
With those variables, the half-shifted function \eqref{eq:ShiftedT}  is written in the form,
\begin{equation*}\label{eq:fR}
f(\varphi;T):=\sum_{m=-\infty}^{\infty}e^{-\frac{1}{2}m(m-1)T+m\varphi}.
\end{equation*}
Noting here that the \emph{triangular} number $\frac{1}{2}m(m-1)$ takes the 
 same value for $m$ and $1-m$,  the sum can be written in the expansion form with the real parameter
 $\mu:=e^{-T}<1$,
 \begin{equation}\label{eq:expansion}
 f(\varphi;T)=\sum_{m=1}^\infty e^{-\frac{1}{2}m(m-1)T}(e^{m\varphi}+e^{(1-m)\varphi}).
 \end{equation}

Now we consider the \emph{soliton limit}, which is given by $\mu\to 0$ (i.e. $T\to+\infty$). The limit corresponds to a singular
limit in the Riemann surface, which may be interpreted as pinching a homological cycle as explained below.
As shown in \eqref{eq:ShiftedT}, in the limit, $f(\varphi; T)$ converges to
$f(\varphi;T)\to 1+e^{\varphi}$.
Then an elliptic function $u(\varphi;T)$ given by the second derivative of the theta function gives a soliton form,
\[
u(\varphi;T):=(\log f(\varphi;T))_{\varphi\varphi}\quad \longrightarrow\quad \frac{1}{4}\sech^2\frac{1}{2}\varphi.
\]
In \cite {T:75}, Morikazu Toda proved that 
 even for a finite $T$, the function $u(\varphi;T)$ can be expressed in the \emph{sum} of the solitons, i.e.
\begin{equation}\label{eq:SolitonSum}
u(\varphi;T)=\frac{1}{4}\sum_{m=-\infty}^{\infty}\sech^2\frac{1}{2}(\varphi-mT),
\end{equation}
which shows that $T$ is the period of $u(\varphi,;T)$. 
M. Toda derived the formula \eqref{eq:SolitonSum} using the well-known infinite product expression of $\theta_3(z;\Omega)$,
 \[
 \theta_3(z;\Omega)=q_0\prod_{n=1}^\infty(1+q^{2n-1}e^{2\pi iz})(1+q^{2n-1}e^{-2\pi i z}),
 \]
 where $q^2=e^{2\pi i\Omega}=e^{-T},~q_0=\prod_{n=1}^\infty(1-q^{2n})$ (see e.g. \cite{M:97}).
 In terms of the variables $(\varphi,T)$ with the half shift $z\to z-\frac{1}{2}\Omega$, we have 
 \[
 f(\varphi;T)=q_0(1+e^\varphi)\prod_{n=1}^\infty (1+e^{\varphi-nT})(1+e^{-\varphi-nT})
 \]
Then calculating $u=(\log f)_{\varphi\varphi}$, one can find the formula \eqref{eq:SolitonSum}.

 We remark that the formula \eqref{eq:SolitonSum} of the soliton sum is a consequence of the \emph{quasi-periodicity} of $f(\varphi;T)$ with respect to $\varphi$ as follows. Under the translation $\mathcal{T}_n: \varphi\to \varphi+nT$ for any $n\in\Z$, we have the quasi-periodicity of $f(\varphi;T)$, 
 \begin{equation}\label{eq:quasiP}
 f(\varphi+nT;T)=e^{\frac{1}{2}n(n+1)T+n\varphi}f(\varphi;T).
 \end{equation}
 This implies that the function $u(\varphi;T)=(\log f(\varphi;T))_{\varphi\varphi}$ is periodic in $\varphi$ with the period $T$.
  Note that the translation $\mathcal{T}_n$ forms a group, write $\mathcal{T}=\langle \mathcal{T}_1\rangle$, and we have $\R/\mathcal{T}\cong S^1$, a circle, and the function $u(\varphi;T)$ is defined on $S^1$.
 If we are interested in the function $u(\varphi;T)$ on a finite interval, say $\varphi\in I[M]:=(-(M+\frac{1}{2})T,(M+\frac{1}{2})T)$, then $u(\varphi;T)$ on this interval may be computed by taking a finite sum,
 \begin{equation}\label{eq:TrancRT1}
 f[M](\varphi;T):= \sum_{m=1}^{M+1}e^{-\frac{1}{2}m(m-1)T}(e^{m\varphi}+e^{(1-m)\varphi}).
 \end{equation}
 Note here that $f[0](\varphi;T)=f(\varphi;\infty)$.
It is easy to see that $u[M](\varphi;T):=(\log f[M](\varphi;T))_{\varphi\varphi}$ consists of just $2M+1$ solitons
for all $\varphi\in\R$, and $u(\varphi;T)\approx u[M](\varphi;T)$ for $\varphi\in I[M]$ as shown in Figure \ref{fig:OneSol}, where we take $M=2$ and $T=20$. One can also see that the error in this region
 is just $e^{-\frac{T}{2}}$, i.e. $|u(\varphi;T)-u[M](\varphi;T)|=\mathcal{O}(e^{-\frac{T}{2}})$. In particular, when $\varphi\in I[0]$,
 we have $u(\varphi;T)=u(\varphi;\infty)+\mathcal{O}(e^{-\frac{T}{2}})$, that is, in the interval $I[0]$, $u(\varphi;T)$
 can be well approximated by the soliton, $\sech^2$-shape.
  We extend these results for higher genus case 
in Section \ref{sec:Riemann}.

 \begin{figure}[h]
    \includegraphics[height=4cm,width=7cm]{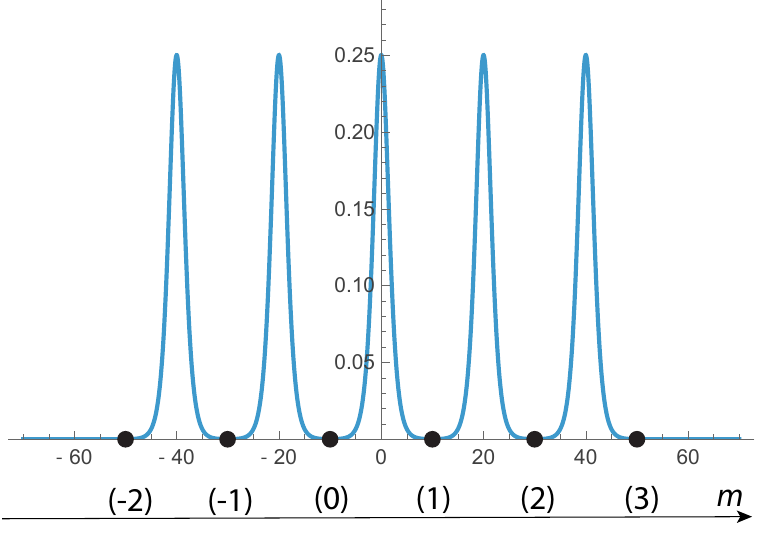}
  \caption{The graph of $u[M](\varphi;T)=(\log f[M](\varphi;T))_{\varphi\varphi}$ with $M=2$ and $T=20$. The labels $(m)$ shows the index $m\in\Z$ in \eqref{eq:TrancRT1}, and the dots on the $\varphi$-axis indicate the corresponding indices. The dominant exponential in \eqref{eq:TrancRT1} is $e^{m\varphi}$ around the dot marked by $(m)$.\label{fig:OneSol}}
\end{figure}

We now explain the soliton limit, $T\to \infty$, as a singular limit of the Riemann surface of $g=1$.
We employ the idea of the Schottky uniformization, which can be used to explain the soliton limits even for higher genus cases. It is known after Poincar\'e and Koebe \cite{Ford, Ba:97} that any compact Riemann surface can be uniformized by the Schottky group. 
Here we give a brief but essential outline of this idea for the case of $g=1$.
The main idea is to consider 
a M\"obius transform $\gamma\in PSL_2(\R)$ (linear fractional transform on $\C\P^1$) with two fixed points $\alpha, \beta\in\R$ assuming $\alpha<\beta$,
\begin{equation}\label{eq:FLT1}
\frac{\gamma(z)-\alpha}{\gamma(z)-\beta}
=\mu\frac{z-\alpha}{z-\beta}, \qquad\text{which gives}\qquad \gamma=\frac{1}{(\beta-\alpha)\sqrt{\mu}}
\begin{pmatrix} \alpha-\mu\beta &-\alpha\beta(1-\mu)\\ 1-\mu & -\beta+\mu\alpha\end{pmatrix}
\end{equation}
where $\mu$ is the multiplier, which is related to the positive real parameter in the expansion \eqref{eq:expansion}, $\mu=e^{-T}$
(see below).
The triplet $\{\alpha,\beta,\mu\}$ is called the Schottky data.
Figure \ref{fig:Torus} illustrates the uniformization. The right panel is obtained by cutting along
the $a$-cycle of the Riemann surface, torus,  given in the left panel.
\begin{figure}[h]
\centering
\includegraphics[width=12cm]{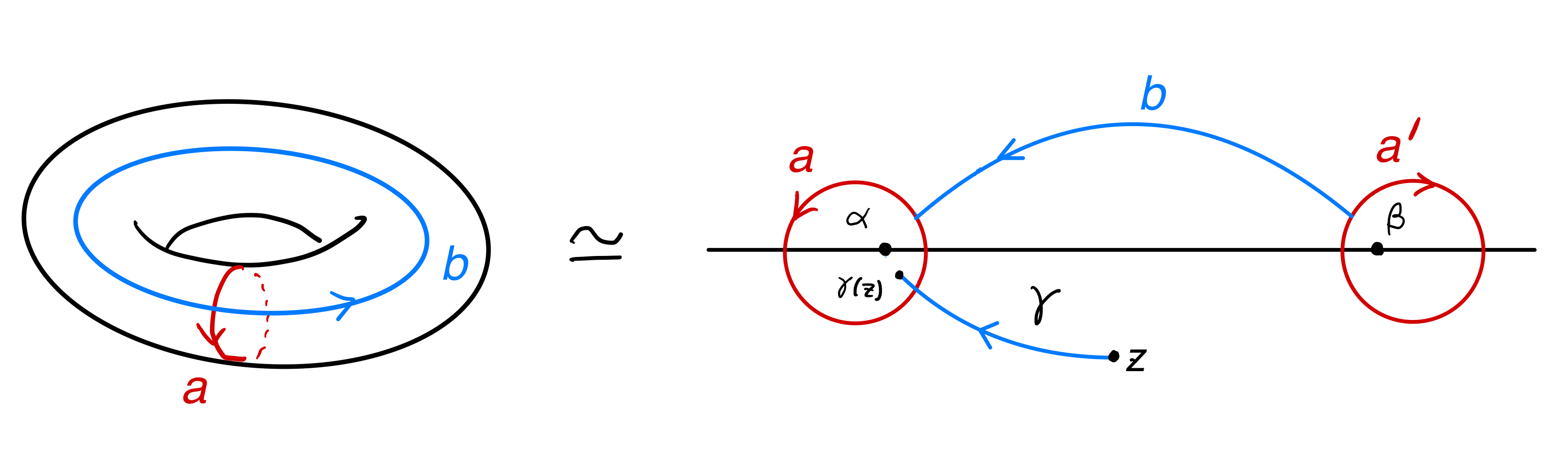}
\caption{Torus with homological cycles $a$ and $b$. Cutting along $a$-cycle, the torus can be viewed as 
$\C\P^1$ without the interior points of the isometric circles $I(\gamma)$ and $I(\gamma^{-1})$.
We identify the $a$-cycle with $I(\gamma^{-1})$ (and $a'$ as $I(\gamma)$).
This illustrates the Schottky uniformization for a torus. \label{fig:Torus}}
\end{figure}
We then identify the $a$-cycle with the isometric circles $I(\gamma)$ and $I(\gamma^{-1})$ of the
M\"obius transform. The isometric circle $I(\gamma)$ of $\gamma$ in \eqref{eq:FLT} is defined by $|d\gamma(z)/dz|=1$ \cite{Ford}, which gives the equation of a circle,
\[
|(1-\mu)z-(\beta-\mu\alpha)|=(\beta-\alpha)\sqrt{\mu}.
\]
The center and radius of $I(\gamma)$ are 
\begin{equation}\label{eq:CR1}
\text{Center}=\frac{\beta-\mu\alpha}{1-\mu},\qquad
\text{Radius}=\frac{\beta-\alpha}{1-\mu}\sqrt{\mu}.
\end{equation}
Note that $\gamma$ maps the outside the isometric circle $I(\gamma)$ into
the interior of $I(\gamma^{-1})$. Thus the Schottky uniformization gives the identification or gluing 
the isometric circles corresponding to the $a$- and $a'$-cycles. Let $\mathcal{F}(\Gamma)$ denote the fundamental domain of the Schottky group, which is defined by
\[
\mathcal{F}(\Gamma):=\text{Ext}\left(\overline{\text{Int}(I(\gamma))}\cup \text{Int}(I(\gamma^{-1}))\right),
\]
where $\text{Ext}(D)$ means the set of exterior points of the set $D$, and $\text{Int}(I(\gamma))$
represents the interior points of the circular disk bounded by $I(\gamma)$.
Then the Schottky uniformization of $\mathcal{R}$ may be expressed as the quotient, $\mathcal{R}\cong \mathcal{F}(\Gamma)/\Gamma$.
On the torus $\mathcal{R}$, there exists an Abelian differential, and the normalized one is given by
\[
\omega(z)=\frac{dz}{2\pi i}\left(\frac{1}{z-\alpha}-\frac{1}{z-\beta}\right)\qquad\text{for}\quad  z\in\mathcal{F}(\Gamma),
\]
from which the moduli $\Omega$ is given by
\[
\Omega=\oint_{b}\omega=\int_c^{\gamma(c)}\omega, \qquad\text{mod}(\Z),
\]
where $c$ is a point on $I(\gamma)$. Calculating the integral explicitly, we have, from \eqref{eq:FLT1},
\[
\Omega=\frac{1}{2\pi i}\log\frac{z-\alpha}{z-\beta}\Big|_{c}^{\gamma(c)}=\frac{1}{2\pi i}\log\frac{(\gamma(c)-\alpha)(c-\beta)}{(\gamma(c)-\beta)(c-\alpha)}=\frac{1}{2\pi i}\log \mu, \quad\text{mod}(\Z).
\]
That is, we have $\mu=\exp (2\pi i\Omega)=e^{-T}<1$. Then the soliton limit $T\to\infty$ is interpreted 
as a ``pinch'' of the $a$-cycle. In terms of the Schottky data,  the soliton limit implies $\mu\to 0$, and the isometric circles $I(\gamma^{-1})$ and $I(\gamma)$ converge to the points $\alpha$ and $\beta$, respectively.

The main purpose of the present paper is to show that the similar structure exists for the higher genus cases.
Then we give some applications of the results to the integrable system. In particular,
we describe some of the \emph{real} and \emph{regular} quasi-periodic (or finite gap) solutions of the KP equation in terms of the Riemann theta function. The KP equation is a prototype equation of the
integrable system, and it is known that the equation admits a large variety of \emph{soliton} solution.
The {real} and {regular} soliton solutions have been classified in terms
of the totally nonnegative Grassmannians (see e.g. a review book \cite{K:17}). Among those KP solitons,
there is a special but important class of solitons, called \emph{Hirota-type} soliton solutions, which can be also obtained by taking singular (or \emph{soliton}) limits of real and regular quasi-periodic solutions (see Section \ref{sec:KP}).  Other types of KP solitons may be obtained by taking some \emph{degenerations} of the soliton parameters
in the Hirota solitons (we explain this later in the present paper, but the details will be announced elsewhere).
In this paper, we start with quasi-periodic solutions, and then consider the soliton limits to find KP solitons.
It should be noted that there are different approaches to study the relation between solitons and quasi-periodic solutions. Notably, in \cite{AG:18b}, they identify a reducible Riemann surface as a deformation of
a singular curve associated with KP soliton, and find that those Riemann surfaces are $\mathsf{M}$-curves
for real and regular quasi-periodic solutions \cite{DN:89}. Our study is obviously related to theirs, but we leave the
detail discussion for a future communication.

 The paper is organized as follows. In Section \ref{sec:RiemannS}, we provide a brief review of the Schottky 
 uniformization of compact Riemann surfaces, and discuss singular Riemann surface obtained by
 pinching a set of homological cycles. This section gives a background information of the present study.
In this paper we consider the \emph{real} Riemann theta function of genus $g$ for near singular surface, and 
as the multi-variable extension of \eqref{eq:phiT}, we
introduce
 the variables $\varphi_p:=2\pi i z_p$, $e^{-T_p}:=\exp(2\pi i \Omega_{p,p})$ and $C_{p,q}:=
 \exp(2\pi i\Omega_{p,q})$ where $(\Omega_{p,q})_{1\le p,q\le g}$ is the $g\times g$ Riemann period matrix. 
  In Section \ref{sec:Riemann}, we show that the quasi-periodicity of the real Riemann theta function
  leads to a \emph{tessellation} of the $g$-dimensional real space $\R^g$ (Proposition \ref{prop:Tessellation}).
  The cell or tile of the tessellation is given by a $g$-dimensional \emph{oblique prism}, which is divided into
  $2^g$ sections by hyper-planes of the dominant exponents in the theta function.
  For example, in the genus two $(g=2)$ case, we have 
 the tessellation with \emph{parallelogram}. We then show that in the soliton limits $T_p\to \infty$
 for all $p=1,\ldots,g$, the Riemann theta function leads to the $M$-theta function introduced in \cite{K:24}.
 We also discuss the cases with partial soliton limits, that is, we take limits $T_p\to\infty$ for some $p$.
 The theta function with a partial pinch generates the soliton solution on quasi-periodic background
 (Proposition \ref{prop:partialPinch}).
 In Section \ref{sec:KP}, we give the details of the case with $g=2$, and study
 the $g=2$ and $3$ quasi-periodic solutions of the KP and KdV equations as the application of the results.
 We also discuss solitons on quasi-periodic background for these equations. In the final section, Section \ref{sec:discussion}, we give a brief summary of the results and discuss
 how one can see KP solitons (Grassmannians) in and from quasi-periodic solutions (Riemann surfaces).


\section{The compact Riemann surface and the Schottky uniformization}\label{sec:RiemannS}
Here we give a brief review of the Riemann surface and the Schottky uniformization as a background information for the present paper.
Let $\mathcal{R}_g$ be a smooth compact Riemann surface of genus $g$.
Let $H_1(\mathcal{R}_g,\mathbb{Z})$ be the first homology group of $\mathcal{R}_g$, and
a set $\{a_1,\ldots,a_g,b_1,\ldots,b_g\}$ be a canonical basis in $H_1(\mathcal{R}_g,\mathbb{Z})$, that is, we have the intersection products,
\begin{equation*}
a_p\circ a_q=0,\quad b_p\circ b_q=0,\quad a_p\circ b_q=\delta_{p,q}.
\end{equation*}
It is well-known that any compact Riemann surface of genus $g$ is homeomorphic to a sphere with
$g$ handles (see e.g. \cite{FK:91}).
The left panel of the figure below shows a Riemann surface of genus 2.
\begin{figure}[h]
  \centering
  \includegraphics[height=4cm,width=11.5cm]{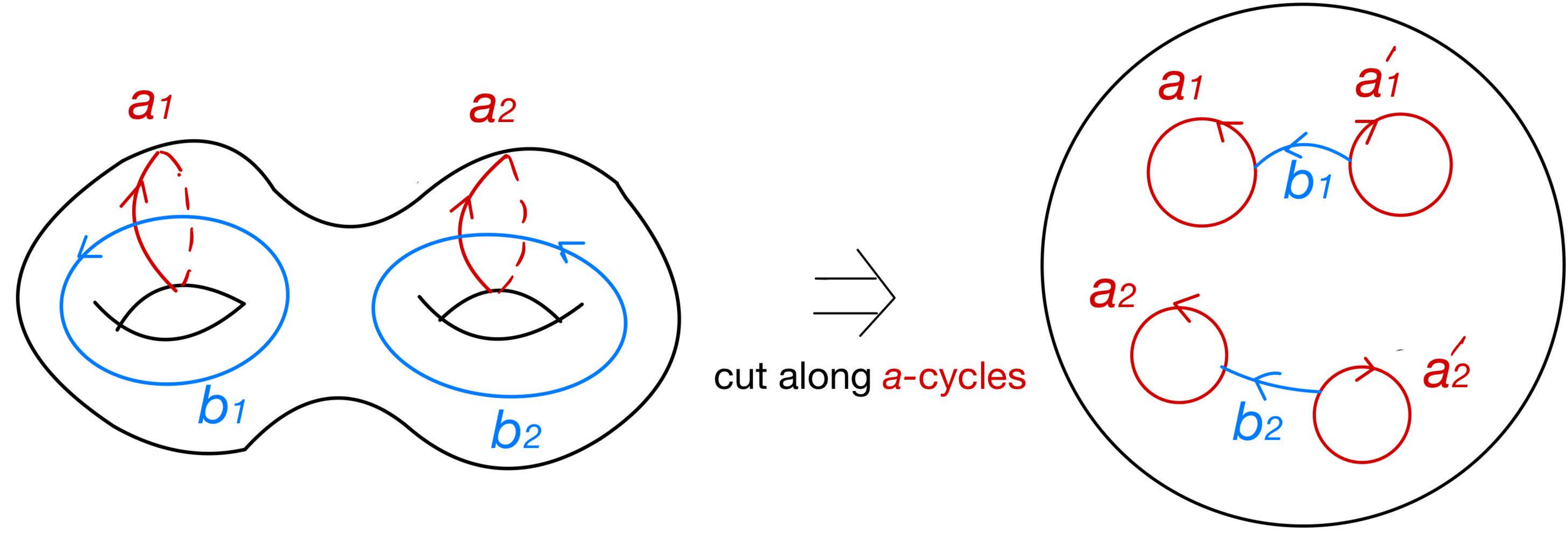}
 \caption{The Riemann surface with genus two. The homological cycles are given by $\{a_p,b_p:p=1,2\}$.
 The right panel illustrates $\C\P^1$ with two holes, which is obtained by cutting along the $a_p$-cycles.
 \label{fig:Riemann2}}
\end{figure}
Cutting the Riemann surface along the $a$-cycles, we obtain the Riemann sphere $\C\mathbb{P}^1$ with
$2g$ holes whose boundaries are marked by $a_p$- and $a'_p$-cycles  as shown in the right panel of the figure. This implies that the Riemann surface can be obtained by identifying each pair of $a_p$- and $a'_p$-cycles. The identification can be expressed by a Schottky group \cite{Ford}.  One should recall that any compact Riemann surface $\RR$ can be uniformized by the Schottky group $\Gamma$, which is defined 
as a finitely generated, discontinuous subgroup of $PSL_2(\C)$ which are free
and purely loxodromic \cite{Ford, Ba:97}. The uniformization is then represented as
\[
\RR~\cong~\Delta(\Gamma)/\Gamma,
\]
where $\Delta(\Gamma)$ is the set of discontinuity of $\Gamma$ (\cite{Ford}, see also \cite{BBEIM:94, Ba:97, Bur1:91, Bur2:91}).

Our Schottky group $\Gamma$ is a discrete subgroup of $PSL_2(\R)$, finitely generated by purely hyperbolic M\"obius transforms on $\C\mathbb{P}^1$ having two fixed points $\{\alpha_p,\beta_p\}$ 
for each transform. Write
$\Gamma=\langle \gamma_p:p=1,\ldots,g\rangle$.  Each element $\gamma_p\in \Gamma$ is given by
\begin{equation}\label{eq:FLT}
\frac{\gamma_p(z)-\alpha_p}{\gamma_p(z)-\beta_p}
=\mu_p\frac{z-\alpha_p}{z-\beta_p}, 
\end{equation}
where $\mu_p$ is the multiplier which is a real constant with $0<\mu_p<1$. We assume 
the fixed points be all \emph{real} with the order $\alpha_p<\beta_p$ for the application to the real and regular solutions of the KP equation.
These fixed points $\alpha_p$ and $\beta_p$ are attractive and repulsive, respectively.
The element $\gamma_p\in PSL_2(\R)$ can be explicitly written in the form,
\begin{equation}\label{eq:gamma}
\gamma_p=\frac{1}{(\beta_p-\alpha_p)\sqrt{\mu_p}}
\begin{pmatrix}
\alpha_p-\mu_p\beta_p & -\alpha_p\beta_p(1-\mu_p)\\
1-\mu_p & -(\beta_p-\mu_p\alpha_p)
\end{pmatrix}.
\end{equation}
The set of $3g$ parameters $\{(\alpha_p,\beta_p,\mu_p):p=1,\ldots,g\}$ is called the \emph{Schottky data}.

The isometric circle $I(\gamma_p)$ of $\gamma_p$ in \eqref{eq:gamma} is defined by $|d\gamma_p(z)/dz|=1$ \cite{Ford}, which gives
\[
|(1-\mu_p)z-(\beta_p-\mu_p\alpha_p)|=(\beta_p-\alpha_p)\sqrt{\mu_p},
\]
whose center and radius are 
\begin{equation}\label{eq:C-R}
\text{Center}=\frac{\beta_p-\mu_p\alpha_p}{1-\mu_p},\qquad
\text{Radius}=\frac{\beta_p-\alpha_p}{1-\mu_p}\sqrt{\mu_p}.
\end{equation}
Two real points on $I(\gamma_p)$, denoted as $\beta_p^{\pm}$, are calculated as
\begin{equation}\label{eq:realPs}
\beta_p^{\pm}(\alpha,\beta)=\frac{\beta_p\mp \alpha_p\sqrt{\mu_p}}{1\mp \sqrt{\mu_p}},\quad\text{which give}\quad \beta_p^-<\beta_p<\beta^+_p.
\end{equation}
Note that the real points on $I(\gamma_p^{-1})$, denoted as $\alpha_p^{\pm}$, are given by
$\alpha_p^{\pm}(\alpha,\beta)=\beta_p^{\mp}(\beta,\alpha)$.
Given the points $\{\alpha_p,\beta_p:p=1,\ldots,g\}$, one can take $\mu_p$ small enough, so that all the isometric circles are \emph{disjoint}.
Recall that $\gamma_p$ maps the outside of the isometric circle $I(\gamma_p)$ into
the interior of $I(\gamma_p^{-1})$. See Figure \ref{fig:Torus}.

The (isometric) fundamental region, denoted by $\mathcal{F}(\Gamma)$, of $\Gamma$ is given by $\C\mathbb{P}^1$ with $2g$ holes of isometric circles, i.e.
\begin{equation}\label{eq:F}
\mathcal{F}(\Gamma):=\text{Ext}\,\left({\bigcup_{p=1}^g\,\overline{\text{Int}\,\left(I(\gamma_p)\right)}\cup \text{Int}\,\left(I(\gamma_p^{-1})\right)}\right),
\end{equation}
where $\text{Ext}({D})$ means the set of exterior points of the set ${D}$, and
$\text{Int}(I(\gamma))$ represents the interior points of the isometric circle $I(\gamma)$.
Then the normalized holomorphic differentials, $\omega_p$  on $\Delta(\Gamma)$, the set of discontinuity of $\Gamma$ can be expressed as \cite{Ford, Ba:97, BBEIM:94}
\begin{equation}\label{eq:Uomega}
\omega_p(z)=\frac{dz}{2\pi i}\sum_{\gamma\in\Gamma/\langle\gamma_p\rangle}
\left(\frac{1}{z-\gamma(\alpha_p)}-\frac{1}{z-\gamma(\beta_p)}\right)\qquad \text{on}\quad \Delta(\Gamma),
\end{equation}
where $\gamma$ runs through all representatives of the right coset classes of $\Gamma$ by its cyclic subgroup $\langle\gamma_p\rangle$ generated by $\gamma_p$. Here $\Delta(\Gamma)$ can be expressed as $\Delta(\Gamma)=\cup_{\gamma\in\Gamma}\gamma(\mathcal{F}(\Gamma))$. 
It is also known \cite{Ford, Ba:97, Sch:13} that the infinite sum in \eqref{eq:Uomega} converges absolutely for sufficiently small $\mu_p$. 
Then we have the following proposition on the period matrix.
The proof can be found in \cite{Ba:97, Bur1:91, Bur2:91,BB:89, I:23, IK:25}. 

\begin{proposition}\label{prop:PeriodI}
The period matrix $\Omega=(\Omega_{p,q})$ of the Riemann surface uniformized by the Schottky group is calculated as
\begin{align*}
\Omega_{p,p}&=\oint_{b_p}\omega_p=\frac{1}{2\pi i}\left(\ln\mu_p+\sum_{\gamma\in\langle\gamma_p\rangle\setminus\Gamma/\langle\gamma_{p}\rangle, \gamma\ne Id}\ln\,[\alpha_p,\beta_p;\gamma(\alpha_p),\gamma(\beta_p)]\right),\qquad {\rm mod}(\Z),\\
\Omega_{p,q}&=\oint_{b_p}\omega_{q}=\frac{1}{2\pi i}\sum_{\gamma\in\langle\gamma_p\rangle\setminus\Gamma/\langle\gamma_{q}\rangle}\ln\,[\alpha_p,\beta_p;\gamma(\alpha_q),\gamma(\beta_q)], \qquad(p\ne q),\qquad {\rm mod}(\Z), \label{eq:Rperiod}
\end{align*}
where $[\alpha_p,\beta_p;\gamma(\alpha_q),\gamma(\beta_q)]$ is the cross ratio given by
\[
[\alpha_p,\beta_p;\gamma(\alpha_q),\gamma(\beta_q)]:=\frac{(\alpha_p-\gamma(\alpha_q))(\beta_p-\gamma(\beta_q))}{(\alpha_p-\gamma(\beta_q))(\beta_p-\gamma(\alpha_q))}.
\]
\end{proposition}

\begin{remark}\label{rem:generalProblem}
Given an algebraic curve, finding a canonical homological basis is quite difficult in general, hence one may not be able to find an explicit form of the period matrix. It seems that most of the previous results are for hyper-elliptic cases, and almost no result for the cases with more general algebraic curves. One should  note that our method based on the Schottky
uniformization does not require any \emph{explicit} form of the algebraic curve in a polynomial form.
\end{remark}

We introduce new variables $T_p$ and $C_{p,q}$, and write the elements of the period matrix as
\begin{align}\label{eq:TC}
e^{-T_p}:=&e^{2\pi i\Omega_{p,p}}=\mu_p\prod_{\gamma\in\langle\gamma_p\rangle\setminus\Gamma/\langle\gamma_{p}\rangle,\gamma\ne Id}[\alpha_p,\beta_p;\gamma(\alpha_p),\gamma(\beta_p)],\\
 C_{p,q}:=&e^{2\pi i\Omega_{p,q}}=\prod_{\gamma\in\langle\gamma_p\rangle\setminus\Gamma/\langle\gamma_{q}\rangle}
[\alpha_p,\beta_p;\gamma(\alpha_q),\gamma(\beta_q)],\quad (p\ne q).\nonumber
\end{align}

In this paper, we are interested in the case with $T_p\gg 0$, and
the limit $T_p\to \infty$ for all $p$, which we call
 the \emph{soliton limit}. Now it is clear that these limits imply pinching of all $a$-cycles, i.e.
the radius of the isometric circles vanish in the limits, i.e. all $\mu_p\to 0$. 
For a small value of $\mu_p$, $\gamma_p(z)$ in \eqref{eq:FLT} can be evaluated as
\begin{equation}\label{eq:gammaEst}
\left\{\begin{array}{lll}
\gamma_p(z)=\displaystyle{\alpha_p+\mu_p(\alpha_p-\beta_p)\frac{z-\alpha_p}{z-\beta_p-\mu_p(z-\alpha_p)}}=\alpha_p+\mathcal{O}(\mu),\\
\gamma^{-1}_p(z)=\displaystyle{\beta_p+\mu_p(\beta_p-\alpha_p)\frac{z-\beta_p}{z-\alpha_p-\mu_p(z-\beta_p)}}=\beta_p+\mathcal{O}(\mu),
\end{array}\right.
\end{equation}
where $\mathcal{O}(\mu)$ is the standard order symbol with $\mu=\text{Max}\{\mu_p:p=1,\ldots,g\}$.
Then we have the following proposition.
\begin{proposition}\label{prop:Estimate}
For a small $\mu\ll 1$, we have the estimates,
\begin{equation*}
e^{-T_p}=\mu_p(1+\mathcal{O}(\mu)),\qquad\text{and}\qquad
C_{p,q}=[\alpha_p,\beta_p;\alpha_q,\beta_q](1+\mathcal{O}(\mu)).
\end{equation*}
\end{proposition}
\begin{proof}
Using \eqref{eq:gammaEst}, the cross ratio in \eqref{eq:TC} is calculated as
\[
[\alpha_p,\beta_p;\gamma(\alpha_p),\gamma(\beta_p)]=\left\{
\begin{array}{lll}
\left[\alpha_p,\beta_p;\alpha_q,\alpha_q\right](1+\mathcal{O}(\mu)),\quad\text{if}\quad \gamma=\gamma_q, \\[1.5ex]
\left[\alpha_p,\beta_p;\beta_q,\beta_q\right](1+\mathcal{O}(\mu)),\quad\text{if}\quad \gamma=\gamma^{-1}_q.
\end{array}
\right.
\]
Note that the cross ratios are $[\alpha_p,\beta_p;\alpha_q,\alpha_q]=[\alpha_p,\beta_p;\beta_q,\beta_q]=1$.

Similarly, we have, for $r\ne p, q$ and $\gamma\ne Id$,
\[
[\alpha_p,\beta_p;\gamma(\alpha_q),\gamma(\beta_q)]=\left\{
\begin{array}{lll}
\left[\alpha_p,\beta_p;\alpha_r,\alpha_r\right](1+\mathcal{O}(\mu)),\quad\text{if}\quad \gamma=\gamma_r, \\[1.5ex]
\left[\alpha_p,\beta_p;\beta_r,\beta_r\right](1+\mathcal{O}(\mu)),\quad\text{if}\quad \gamma=\gamma^{-1}_r.
\end{array}
\right.
\]
This completes the proof.
\end{proof}

With the period matrix $\Omega$ given in Proposition \ref{prop:PeriodI}, we have the Riemann theta function,\begin{equation}\label{eq:RiemannTg}
\theta({\bf z};\Omega):=\sum_{{\bf m}\in\mathbb{Z}^g}\exp 2\pi i \left(\frac{1}{2}{\bf m}^T\Omega {\bf m}+{\bf m}^T {\bf z}\right),
\end{equation}
for ${\bf z}\in \C^g$, and $\mathbf{m}^T=(m_1,\ldots,m_g)$ is the transpose of the column vector $\mathbf{m}$.
Writing $2\pi i z=\varphi=(\varphi_1,\ldots,\varphi_g)$, the theta function can be expressed in the following form with the variable \eqref{eq:TC} and the half-shifts $z_p\to z_p-\frac{1}{2}\Omega_{p,p}$,
\begin{equation}\label{eq:realRT}
f(\varphi;T,C):=\sum_{{\bf m}\in\Z^g}\exp \left(\sum_{p=1}^g\left(-\frac{1}{2}m_p(m_p-1)T_p+ m_p\varphi_p\right)\right)\prod_{p<q}(C_{p,q})^{m_pm_q}.
\end{equation}
Near the soliton limit, 
Proposition \ref{prop:Estimate} implies that if the cross ratio $[\alpha_p,\beta_p;\alpha_q,\beta_q]>0$, i.e.
$C_{p,q}>0$,
the theta function \eqref{eq:realRT} is positive definite. 
This is a sufficient condition
for the regularity of quasi-periodic solutions for some integrable systems, like KP equation. One should also note that the regularity of quasi-periodic solution is determined by
that of the soliton solution in the soliton limit (see Section \ref{sec:KP}). 

Since near the soliton limit, we have $\mu_P\approx e^{-T_p}$, we sometimes  use $T_0:=\text{Min}\{T_p:p=1,\ldots,g\}$
to express the order $\mathcal{O}(\mu)$ as $\mathcal{O}(e^{-T_0})$.

\subsection{Singular Riemann surface by pinching all $a$-cycles}

Here we describe the singular limit of the Riemann surface by pinching the $a$-cycles.
In \cite{Mu:84}[Chapter 5, p.3.243], Mumford considered the theta function on singular curve. Let $\mathcal{C}$ be a singular curve of (arithmetic) genus $g$, and let $S=\text{Sing}(\mathcal{C})=\{s_1,\ldots,s_g\}\subset \mathcal{C}$.  Assume that the singularities of $\mathcal{C}$ are only ordinary \emph{double} points $s_1,\ldots,s_g$ and that $\mathcal{C}$ has normalization
\[
\pi:\C\P^1~\to~ \mathcal{C}.
\]
We consider that these singular points are generated as the result of pinching, and 
$\pi^{-1}(s_p)=\{\alpha_p,\beta_p\}$ for $p=1,\ldots,g$ given in \eqref{eq:FLT}. This means that $\mathcal{C}$ is just $\mathbb{P}$ with $g$ pairs of points $\{\alpha_p,\beta_p\}$ identified. 
The meaning of ``ordinary double points" is that the points $\{\alpha_p,\beta_p:p=1,\ldots,g\}$ are all distinct.

\begin{figure}[h]
  \centering
  \includegraphics[height=4cm,width=12cm]{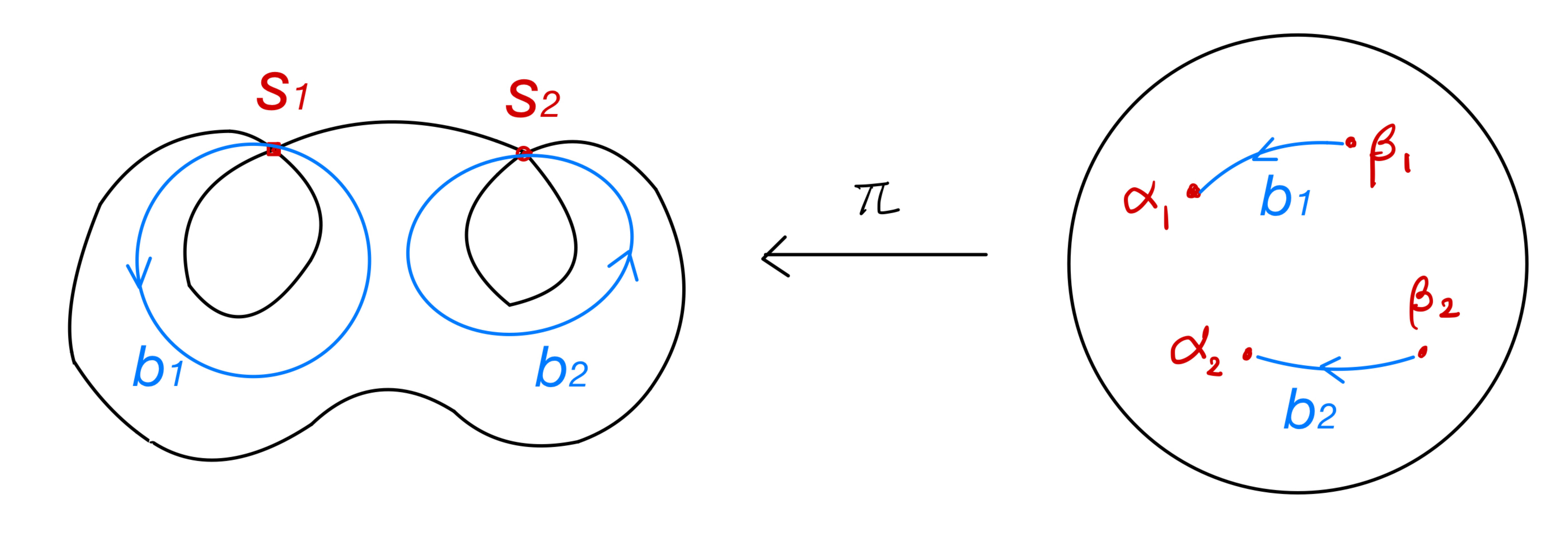}
 \caption{Singular Riemann surface $\mathcal{C}$ obtained by squeezing $a$-cycles in Figure \ref{fig:Riemann2}, and the normalization, $\pi:\C\P^1\to \mathcal{C}$
 with $\pi^{-1}(s_p)=\{\alpha_p,\beta_p\}$. 
 \label{fig:PinchR}}
\end{figure}

Then from Proposition \ref{prop:PeriodI},
the diagonal elements of the period matrix take $\Omega_{p,p}\to i\infty$, i.e.
\begin{equation}\label{eq:limits}
T_{p}~\to~ \infty \quad \text{for}\quad 1\le p\le g,
\end{equation}
and the limits of $C_{p,q}=\exp(2\pi i\Omega_{p,q})$ for $p\ne q$ are given by
\begin{equation}\label{eq:Omega}
C_{p,q}\quad\longrightarrow\quad\tilde{C}_{p,q}:=[\alpha_p,\beta_p;\alpha_q,\beta_q]=\frac{(\alpha_p-\alpha_q)(\beta_p-\beta_q)}{(\alpha_p-\beta_q)(\beta_p-\alpha_q)}.
\end{equation}
One should also note that in the limits \eqref{eq:limits}, the holomorphic differential \eqref{eq:Uomega}
becomes
\begin{equation}\label{eq:Lomega}
\omega_p\quad\longrightarrow\quad \tilde{\omega}_p=\frac{dz}{2\pi i}\left(\frac{1}{z-\alpha_p}-\frac{1}{z-\beta_p}\right)\qquad 
\text{on}\quad \C\P^1\setminus \pi^{-1}(S).
\end{equation}

In the limit $T_p\to \infty$, the Riemann theta function \eqref{eq:realRT} becomes
\begin{align}\label{eq:Mtheta}
f(\varphi;\infty,\tilde{C}):=&\sum_{{\bf m}\in\{0,1\}^g}\prod_{p<q}(\tilde{C}_{p,q})^{m_pm_q}\,\exp\left(\sum_{p=1}^g m_p\varphi_p\right)\\
=&1+\sum_{p=1}^g e^{\varphi_p}+\sum_{p<q}\tilde{C}_{p,q}\,e^{\varphi_p+\varphi_q}+\cdots+\left(\prod_{p<q}\tilde{C}_{p,q}\right) e^{\varphi_1+\cdots+\varphi_g}.\nonumber
\end{align}
Thus the {infinite} sum of exponential terms in the theta function \eqref{eq:realRT} becomes a \emph{finite} sum of exponential terms, which is referred to as the \emph{$M$-theta function} in \cite{K:24}. We refer to the limit as the \emph{soliton limit}, and to
the set of $2g$ parameters $\{(\alpha_p,\beta_p):p=1,\ldots,g\}$ as the \emph{soliton parameter}.

One should also note that $C_{p,q}=\tilde{C}_{p,q}(1+\mathcal{O}(e^{-T_0}))$ for large $T_p$. This implies that
if $\tilde{C}_{p,q}>0$, then the theta function \eqref{eq:realRT} does not vanish on the real space $\R^g$.
With $\alpha_p<\beta_p$ and $\alpha_p<\alpha_q$ for $p<q$, the condition $\tilde{C}_{p,q}>0$ means that the set of  pairs 
$\{(\alpha_p,\beta_p):p=1,\ldots,g\}$ satisfies, for $p<q$,
\begin{equation}\label{noncrossing}
\alpha_p<\beta_p<\alpha_q<\beta_q,\qquad\text{or}\qquad \alpha_p<\alpha_q<\beta_q<\beta_p.
\end{equation}
This is nothing but the set of non-crossing pairs of $2g$ points on $\R$, and
 the number of those sets is known to be the Catalan number (see e.g. \cite{S:97}, also \cite{CK:08})
\[
C_g=\frac{1}{g+1}\binom{2g}{g}
\]
In terms of the Schottky uniformization, the condition \eqref{noncrossing} implies that the $b$-cycles
can be all in the upper-half plane of $\C$.  Figure \ref{fig:genus2} illustrates the Schottky groups for  $g=2$, where we have two different types of the Schottky group. They give two different types of quasi-periodic solutions, and the difference is important when we
take the solitons limits (see below).

\begin{figure}[h]
\centering
\includegraphics[width=12cm]{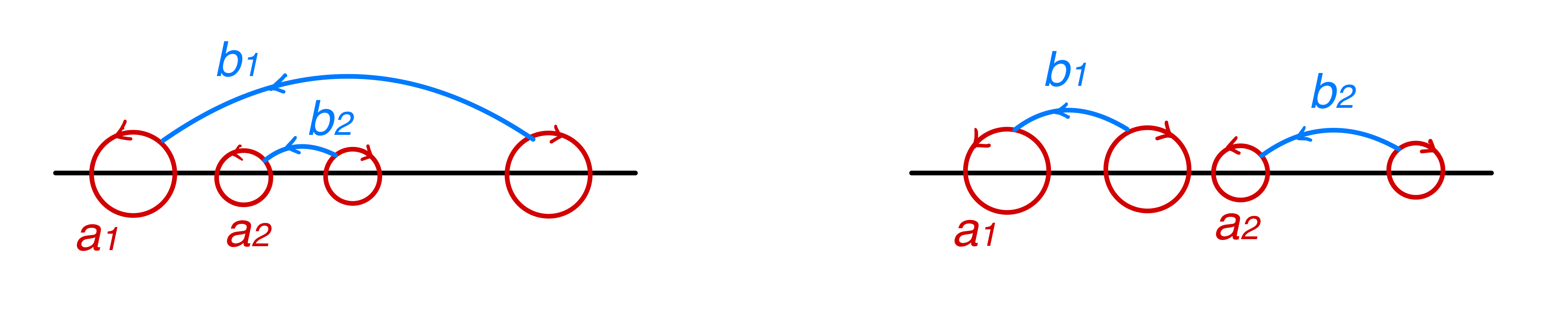}
\caption{Two different types of the Schottky group for $g=2$. The case in the left panel gives $\tilde{C}_{1,2}<1$, and
the right one gives $\tilde{C}_{1,2}>1$.\label{fig:genus2}}
\end{figure}

\begin{remark}\label{rem:Off}
In the singular limit (pinching $a$-cycles),  the diagonal elements $\Omega_{p,p}$ takes the limit,
$\text{Im}(\Omega_{p,p})\to +\infty$. Some of the off diagonal elements $\Omega_{p,q}, \,(p\ne q)$ may also take
$\text{Im}(\Omega_{p,q})\to \pm\infty$ in the limit. However, one can see that $|\Omega_{p,q}|$ is bounded, if $|\Omega_{p,p}|$ is bounded. This is due to the fact $\text{Im}(\Omega)>0$, which means
$\text{Im}({\bf v}^T\Omega {\bf v})>0$  for any $g$-dimensional real vector ${\bf v}\in \R^g$.
Taking ${\bf v}^T=(0,\ldots,0,v_p,0,\ldots,0,v_q,0,\ldots,0)$, 
the condition $\text{Im}(\Omega)>0$ leads to 
\[
T_pv_p^2+T_qv_q^2-2(\log C_{p,q}) v_p v_q>0,
\]
which implies $|\log C_{p,q}|<\sqrt{T_pT_q}$.
Note that the elements $C_{p,q}$ are determined by the Schottky data, and we consider the data such that they satisfy 
\begin{equation}\label{eq:C<T}
|\log C_{p,q}|<T_0:=\text{Min}\{T_p:p=1,\ldots,g\}.
\end{equation}
The condition \eqref{eq:C<T} is important for the \emph{resonance} in the KP solitons, which appears in the case with some $|\log C_{p,q}|=\infty$ (see Section \ref{sec:KP}, and also \cite{K:24} for many examples of the KP solitons with $C_{p,q}=0$).
 It should be noted that $|\log C_{p,q}|=\infty$ implies that the ordinary double
 point  further degenerates to an ordinary \emph{n-tuple} singularity (Nakayashiki
 pointed out to the author). For example, $\alpha_1=\alpha_2$ gives $C_{1,2}=0$ 
 (see \cite{K:24}, and also Section \ref{sec:KP}).
  We also recall that this type of singularity appears in \cite{KX:21} in the study of
 complex KP solitons.
The details will be discussed in a future communication.
\end{remark}


\section{The Riemann theta function and the tessellations of $\R^g$}\label{sec:Riemann}
Here we show that the Riemann theta function \eqref{eq:realRT} for $g>1$ holds the similar properties as those in the elliptic case of $g=1$ discussed in Introduction. Let us write \eqref{eq:realRT} in the form,
\begin{equation}\label{eq:realRT2}
f(\varphi;T,C):=\sum_{{\bf m}\in\Z^g}\exp \sum_{p=1}^g\left(-\frac{1}{2}m_p(m_p-1)T_p+ m_p\varphi_p+\sum_{q>p}m_pm_q\log C_{p,q}\right)
\end{equation}
Then we have the following lemma, which is the \emph{quasi-periodicity} of the theta function $f(\varphi;T,C)$.
\begin{lemma}\label{lem:Qperiodic}
Given $T_p>0$, $C_{p,q}>0$ and $n_p\in\Z$, consider the translations $\mathcal{T}:\R^g\to\R^g$
\[
\mathcal{T}\,:\,\varphi_p\, \longrightarrow\, \varphi_p+n_pT_p-\sum_{q\ne p}n_q\log C_{p,q},\qquad\text{for}\quad p=1,\ldots,g.
\]
Then we have
\begin{equation}\label{eq:Qperiodic}
f(\mathcal{T}(\varphi);T,C)=
\prod_{p=1}^g e^{\frac{1}{2}n_p(n_p+1)T_p}\prod_{p<q}(C_{p,q})^{-n_pn_q}e^{\sum_{p=1}^gn_p\varphi_p}\,f(\varphi;T,C).
\end{equation}
\end{lemma}
\begin{proof}
The proof is just a substitution (e.g. see also \cite{Mu:83}).
\end{proof}

 We then see that second derivative of $\log f(\varphi;T,C)$ is invariant under the translation $\mathcal{T}$. For example, one can take the second derivative in the symmetric form,
 \begin{equation}\label{eq:2Diff}
 u(\varphi;T,C):=\left(\sum_{p=1}^g\frac{\partial}{\partial \varphi_p}\right)^2\log f(\varphi;T,C).
 \end{equation}
The group $\mathcal{T}$ is given by
\begin{equation}\label{eq:Tgroup}
\mathcal{T}=\langle T_{0,\ldots,0,\underset{p}{1},0,\ldots,0}:p=1,\ldots,g\rangle,
\end{equation}
where the generator $T_{0,\ldots,0,\underset{p}{1},0,\ldots,0}$ ($1$ is at the $p$-th place) is defined by
\[
(T_{0,\ldots,0,\underset{q}{1},0,\ldots,0}(\varphi))_p=\varphi_p+\delta_{p,q}T_q-(1-\delta_{p,q})\log C_{p,q},\qquad p=1,\ldots,g.
\]
We assume that $T_p, T_q>|\log C_{p,q}|$ (see Remark \ref{rem:Off}). 

From Lemma \ref{lem:Qperiodic}, the following proposition is immediate.
\begin{proposition}\label{prop:Tessellation}
The group $\mathcal{T}$ in \eqref{eq:Tgroup} gives $g$-dimensional 
tessellation of $\R^g$, and the fundamental domain $\mathcal{F}(\mathcal{T})$ is 
an \emph{oblique prism}, $g$-dimensional cell or tile of the tessellation.
The cell of prism is expressed by the convex hull of $2^g$ points given by
\begin{equation}\label{eq:DomExp}
\{~0,~{\bf v}_p, ~{\bf v}_{p,q}:={\bf v}_p+{\bf v}_q~ (p\ne q),~\ldots,~{\bf v}_{1,\ldots,g}:=\sum_{p=1}^g{\bf v}_p~\}
\end{equation}
where ${\bf v}_p$ is given by
\[
{\bf v}_p=T_p\,{\bf e}_p-\sum_{q\ne p}\log C_{p,q}\,{\bf e}_q,\qquad p=1,\ldots,g,
\]
with the unit vector ${\bf e}_p$ in the $\varphi_p$-direction. Moreover, 
identifying the opposite faces of the cell under the translation, we obtain a real $g$-dimensional torus, called \emph{$g$-torus}, i.e.
\[
\R^g/\mathcal{T}\cong S^1\times\cdots\times S^1.
\]
\end{proposition}
Since the leading order term of $f(\varphi;T,C)$ in the cell is just $f(\varphi;\infty,C)$, the second derivative 
$u(\varphi;T,C)$ in \eqref{eq:2Diff} can be well approximated by the sum of $u(\varphi;\infty,C)$
on the cell over the tessellation.
 For example, in the case with $g=2$, the cell of the tessellation is a parallelogram,
on which $u(\varphi;T,C)$ is given by the second derivative of $\log f(\varphi;\infty,C)$.
Figure \ref{fig:Tess-g2} shows the tiles of the tessellation of $\R^2$. The left panel is for the case with
the left choice of the homological cycles (i.e. $C_{1,2}<1$) in Figure \ref{fig:genus2}, and the right one is for the right choice in Figure \ref{fig:genus2}. In the figure, the solid line shows the peak trajectory of the function $u(\varphi;T,C)$, which forms a hexagon and is dual to the hexagon cell (see Figure \ref{fig:Tess-g2}).
\begin{figure}[h]
\centering
\includegraphics[width=8cm]{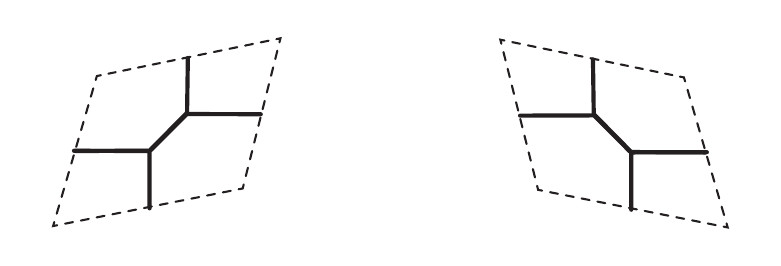}
\caption{The tiles of the tessellations of $\R^2$, parallelograms in dashed lines. The left panel is for $C_{1,2}<1$, and
the right one is for $C_{1,2}>1$. The solid lines show the peak trajectory of $u(\varphi;T,C)$. Those tessellations are related to the choice of the homological 
cycles shown in Figure \ref{fig:genus2}. \label{fig:Tess-g2}}
\end{figure}

One should note that the tessellation depends on the elements of the period matrix $C_{p,q}$, which is 
determined by the choice of the homological cycles.
Figure \ref{fig:Tess} shows the tessellation of $\R^2$ for the case with $C_{1,2}<1$.
In the figure, the hexagon labeled with $(m_1,m_2)$ indicates the region, where
the corresponding exponential term in \eqref{eq:realRT} gives the dominant one.
For example, the term 1 in $f(\varphi;T,C)$ dominates in the region with $(0,0)$, the term $e^{\varphi_1}$ dominates in $(1,0)$ region),
and $C_{1,2}e^{\varphi_1+\varphi_2}$ dominates in $(1,1)$ region.
\begin{figure}[h]
\centering
\includegraphics[width=9cm]{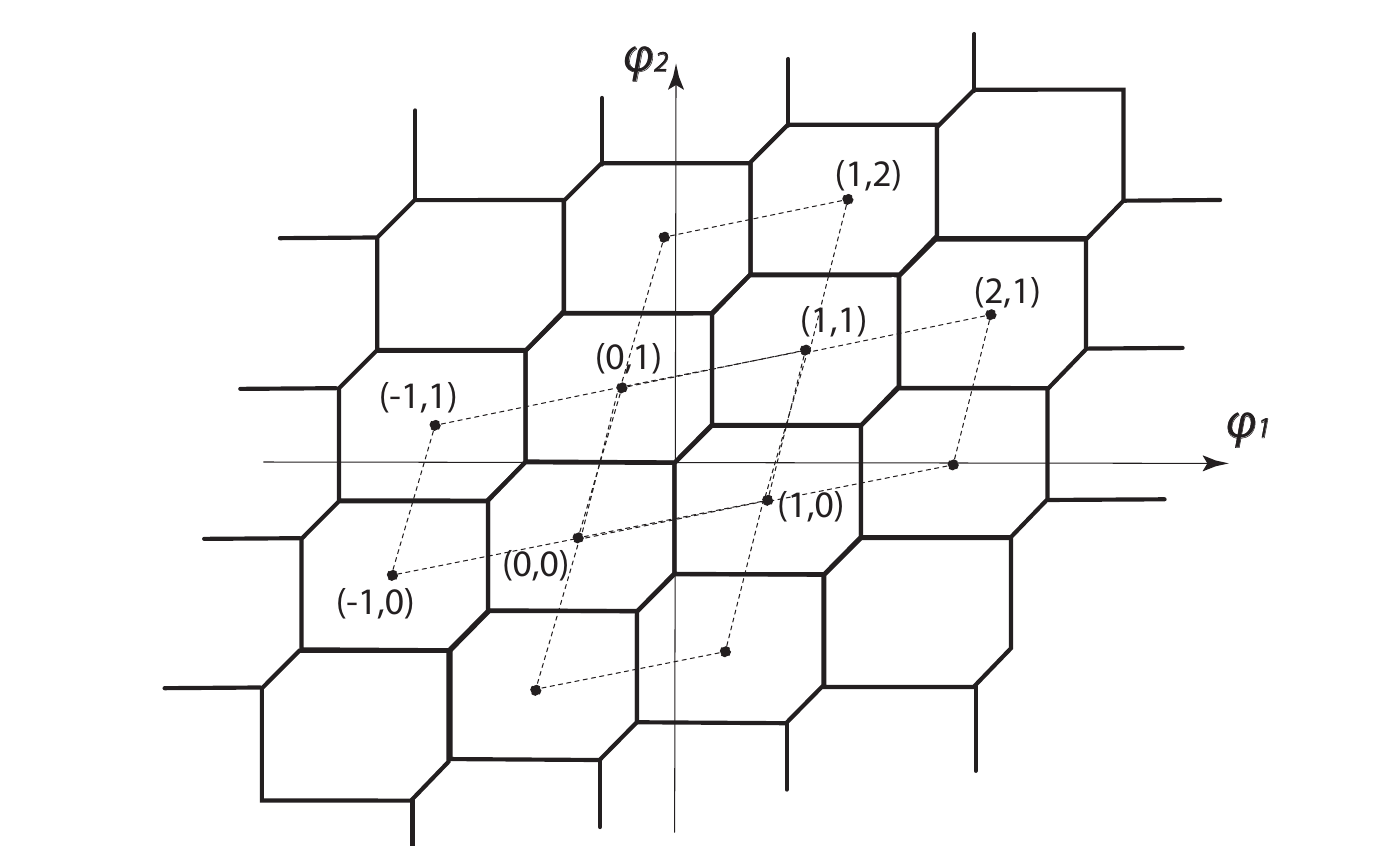}
\caption{The tessellation generated by the $g=2$ theta functions with $C_{1,2}<1$. 
The point marked by $(m_1,m_2)\in\Z^2$ represents the vector $m_1{\bf v}_1+m_2{\bf v}_2+{\bf  v}_0$ with
${\bf v}_p=T_p{\bf e}_p-\log C_{1,2}{\bf e}_{3-p}$ for $p=1,2$. The constant vector ${\bf v}_0$
fixes the origin. \label{fig:Tess}}
\end{figure}

The profile of the function $u(\varphi;T,C)$ in \eqref{eq:2Diff} for these cases are shown in Figure \ref{fig:FPO}.
Notice that the interaction properties are different. The left one shows a depression, and the right one shows an elevation at the intersection points.
\begin{figure}[h]
\centering
\includegraphics[width=11cm]{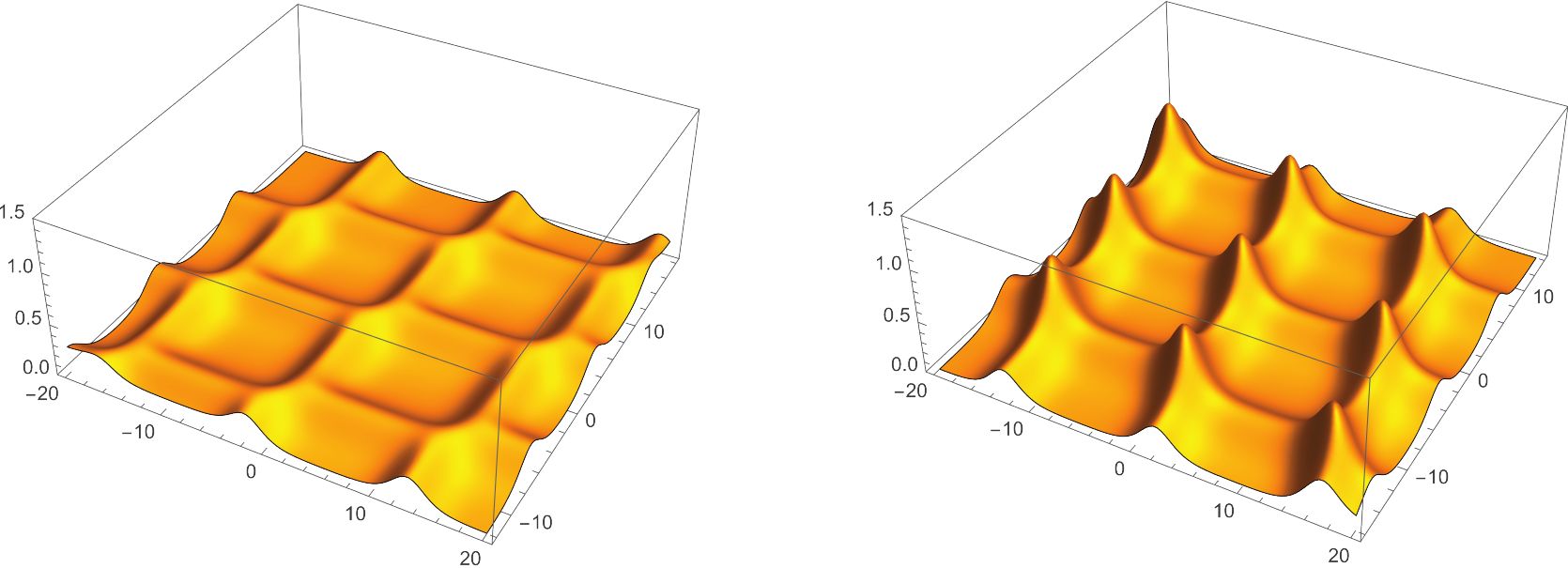}
\caption{The 3-dimensional plots of the $g=2$ quasi-periodic function $u(\varphi;T,C)$ with $T_1=15$ and $T_2=10$. The left one is for $ C_{1,2}<1$, and the right one is for $C_{1,2}>1$.\label{fig:FPO}}
\end{figure}

In general, the dominant exponential in the theta function \eqref{eq:realRT} can be found as follows.
Let $\Phi_{\bf m}(\varphi;T,C)$ for ${\bf m}=(m_1,\ldots,m_g)\in \Z^g$ denote the exponent given by
\begin{equation}\label{eq:dominant}
\Phi_{\bf m}(\varphi;T,C):=\sum_{p=1}^g\left(-\frac{1}{2}m_p(m_p-1)T_p+m_p\varphi_p+\sum_{q>p}m_pm_q\log C_{p,q}\right).
\end{equation}
Then consider the set of hyper-planes in $\R^{g+1}$ given by
$z=\Phi_{\bf m}(\varphi;T,C)$ in \eqref{eq:dominant}, and identify the maximum hyperplane in each point in 
$\varphi\in\R^g$, i.e.
\[
\mathcal{M}(\varphi):=\underset{{\bf m}\in\Z^g}{\text{Max}}\,\{z=\Phi_{\bf m}(\varphi;T,C)\}.
\]
That is, the dominant exponential in the theta function \eqref{eq:realRT} at the point $\varphi\in\R^g$
is given by $\mathcal{M}(\varphi)$. One should note that the dominant exponentials are marked by
the $2^g$ vertices of the cell (oblique prism) in \eqref{eq:DomExp}.
One can define a graph (tessellation pattern) showing the locus of changing the maximum
hyperplane (i.e. the dominant exponential), which is the soliton graph defined in \cite{KW:14}.
Then each cell is divided by the maximum hyperplanes, and the pattern is the same for each cell in $\R^g$.
In Figure \ref{fig:Tess}, the pattern (soliton graph) given in solid lines forms a tessellation of $\R^2$ with hexagon, and notice that the parallelogram in dotted lines is \emph{dual} to the hexagon.

\subsection{The truncation of the theta function}

In Introduction, we explained that the truncated series \eqref{eq:TrancRT1} gives a good approximation
for the infinite sum $f(\varphi;T)$ in \eqref{eq:expansion} in the finite interval 
$I[M]=(-(M+\frac{1}{2})T, (M+\frac{1}{2})T)$. This is also true for the case with higher genus.
 Let ${\bf M}=(M_1,\ldots,M_g)\in(\Z_{\ge 0})^g$, and
 $L[{\bf M}]$ be a subset of points in $\Z^g$ forming a hyper-cube,
\begin{equation}\label{eq:Ldomain}
L[{\bf M}]:=\{{\bf m}=(m_1,\ldots,m_g)\in\Z^g:-M_p\le m_p \le M_p+1\}.
\end{equation}
The total number of the points in $L[{\bf M}]$ is $\prod_{p=1}^g(2M_p+2)$.
Then we define the truncated theta function with $L[{\bf M}]$,
\begin{equation}\label{eq:TrancRTg}
f[{\bf M}](\varphi;T,C)=\sum_{{\bf m}\in L[{\bf M}]}
\exp \Phi_{\bf m}(\varphi;T,C).
\end{equation}
The second derivative $u(\varphi;T,C)$ in \eqref{eq:2Diff} with $f[{\bf M}](\varphi;T,C)$, denoted as
$u[{\bf M}](\varphi;T,C)$, 
 has $2M_p+1$ (line) solitons ($\sech^2$-shape) in each $\varphi_p$-direction for $p=1,\ldots,g$.
Then one can see that $u[{\bf M}](\varphi;T,C)$  gives a good approximation for $u(\varphi;T,C)$
in the domain consisting of the oblique prisms in the tessellation of $\R^g$ with the size,
\[
(2M_1+1)T_1\times\cdots\times (2M_g+1)T_g.
\]
Note that each prism (tile) is labeled by ${\bf m}\in L[{\bf M}]$, and the size of the prism is
of order $T_1\times\cdots\times T_g$. The error in the tile may be estimated as
$|u(\varphi;T,C)-u[{\bf M}](\varphi;T,C)|=\mathcal{O}(e^{-\frac{T_0}{2}})$ with $T_0=\text{Min}\{T_p:p=1,\ldots,g\}$.
Figure \ref{fig:Tess} shows the case with $(M_1=2,M_2=2)$, that is, we have $25=(2M_1+1)\times(2M_2+1) $ parallelograms and $5=2M_p+1$ vertical ($p=1$) and $5$ horizontal ($p=2$) line solitons in both $\varphi_p$-directions.

One should also note that we have $f[{\bf 0}](\varphi;T,C)=f(\varphi;\infty,C)$, which is the theta function in the soliton limit and consists of $2^g$ exponential terms
(see \eqref{eq:Mtheta} with $C=\tilde{C}$).

\subsection{The Riemann theta function for the surface with partial pinch}
In the soliton limits $T_p\to\infty$ for all $p=1,\ldots,g$ (total pinch), the theta function in \eqref{eq:realRT} becomes $f(\varphi;\infty,\tilde{C})$ in \eqref{eq:Mtheta}.
Here we consider the cases when we pinch only some of the $a$-cycles. A formula of the theta function
for these cases has been obtained in the previous paper \cite{K:24}, which can be stated in
the following proposition with our present notation. 

\begin{proposition}\label{prop:partialPinch}
Let $f_g^{(n)}(\varphi;T,\tilde{C})$ be the theta function with partial pinch, $T_p\to \infty$ for
$p=1,\ldots,n<g$. Then we have
\begin{align*}
f_g^{(n)}(\varphi;T,\tilde{C})=&\lim_{T_1,\ldots,T_n\to\infty}f_g(\varphi;T,C)\\
=&f_{g-n}(\varphi^{(n)};T^{(n)},C^{(n)})+\sum_{j=1}^ne^{\varphi_j}
f_{g-n}(\varphi^{(n)}+\delta^{(n)}_j;T^{(n)},C^{(n)})\\
&+\sum_{1\le j<k\le n}e^{\varphi_j+\varphi_k}\tilde{C}_{j,k}f_{g-n}(\varphi^{(n)}+\delta_j^{(n)}+\delta_k^{(n)};T^{(n)},C^{(n)})\\
&+\cdots +e^{\sum_{j=1}^n\varphi_j}\left(\prod_{j<k}\tilde{C}_{j,k}\right)\,f_{g-n}(\varphi^{(n)}+\sum_{j=1}^n\delta_j^{(n)};T^{(n)},C^{(n)}),
\end{align*}
where $f_{g-n}(\varphi^{(n)};T^{(n)},C^{(n)})$ is the genus $g-n$ theta function, and the arguments are defined as
\[\left\{\begin{array}{lll}
\varphi^{(n)}=(\varphi_{n+1},\ldots,\varphi_g),\qquad \varphi^{(n)}+\delta^{(n)}_j=(\varphi_{n+1}+\log \tilde{C}_{j,n+1},\ldots,\varphi_g+\log\tilde{C}_{j,g}),\\[1.5ex]
 T^{(n)}=(T_{n+1},\ldots, T_g),\qquad C^{(n)}=(\tilde{C}_{k,l})_{n+1\le k, l\le g}\quad\text{with}\quad \tilde{C}_{k,l}=\exp(2\pi i\tilde{\Omega}_{k,l}).
 \end{array}\right.
 \]
\end{proposition}
\begin{proof}
Recall that in each limit $T_p\to\infty$, only two cases with $m_p=0, 1$ contribute the sum in the theta function \eqref{eq:realRT}.
The first term in $f_g^{(n)}(\varphi;T,\tilde{C})$ is then given by the choice $m_p=0$ for all $p=1,\ldots,n$. 
Taking $m_p=\delta_{j,p}$ for $1\le j,p\le n$, the exponent in \eqref{eq:realRT} is
\[
\varphi_j\,+\,\sum_{q=n+1}^g\left(-\frac{1}{2}m_q(m_q-1)T_q+m_q(\varphi_q+\log \tilde{C}_{j,q})\right)+\sum_{n+1\le p<q\le g}m_p m_q\log\tilde{C}_{p,q},
\]
which gives the second term for $1\le j\le n$. The third term is then given by the choice with
two nonzero terms $m_p=\delta_{j,p}+\delta_{k,p}$. The exponent is then given by replacing
the variables in the above as follows,
\[
\varphi_j~\to~\varphi_j+\varphi_k+\log \tilde{C}_{j,k},\quad\text{and}\quad \varphi_q+\log \tilde{C}_{j,q}~\to~
\varphi_q+\log \tilde{C}_{j,q}+\log \tilde{C}_{k,q}.
\]
The rest of the proof should be clear.
\end{proof}

\begin{remark}\label{rem:Vertex}
As was mentioned in \cite{K:24}, the formula in Proposition \ref{prop:partialPinch} can be expressed
in an elegant form with the vertex operators,
\[
f^{(n)}_g(\varphi;T,\tilde{C})=\prod_{j=1}^n\left(1+V_j^{(n)}[\tilde{C}]\right)\cdot f_{g-n}(\varphi^{(n)};T^{(n)},C^{(n)}),
\]
where the vertex operator $V_j^{(n)}[\tilde{C}]$ is defined by
\[
V_j^{(n)}[\tilde{C}]:=\exp\left(\varphi_j+\sum_{k=1,k\ne j}^n\log \tilde{C}_{j,k}\frac{\partial}{\partial\varphi_k}+\sum_{l=n+1}^g\log\tilde{C}_{j,l}\frac{\partial}{\partial\varphi_l}\right)\quad\text{for}\quad j=1,\ldots,n.
\]
Nakayashiki has obtained a similar formula \cite{N2:23}. With this formula, the phase shifts between soliton and quasi-periodic background are quite visible (see Figure \ref{fig:CE-2} below).
\end{remark}

Figure \ref{fig:CE-2} shows the case of $g=2$. Here we compute 
$u(\varphi;T,C)$ in \eqref{eq:2Diff}.
The left panel shows
the function $u(\varphi;T,C)$ with $T_1=15, T_2=10$ and $\log C_{1,2}= -2$. The right one shows $u(\varphi;T,\tilde{C})$ with $T_1=\infty, T_2=10$ and $\log\tilde{C}_{1,2}\approx -2$ (the difference is of order $e^{-10}$ and ignored), which is a partial soliton limit of the left one. 
\begin{figure}[h]
\centering
\includegraphics[width=12cm]{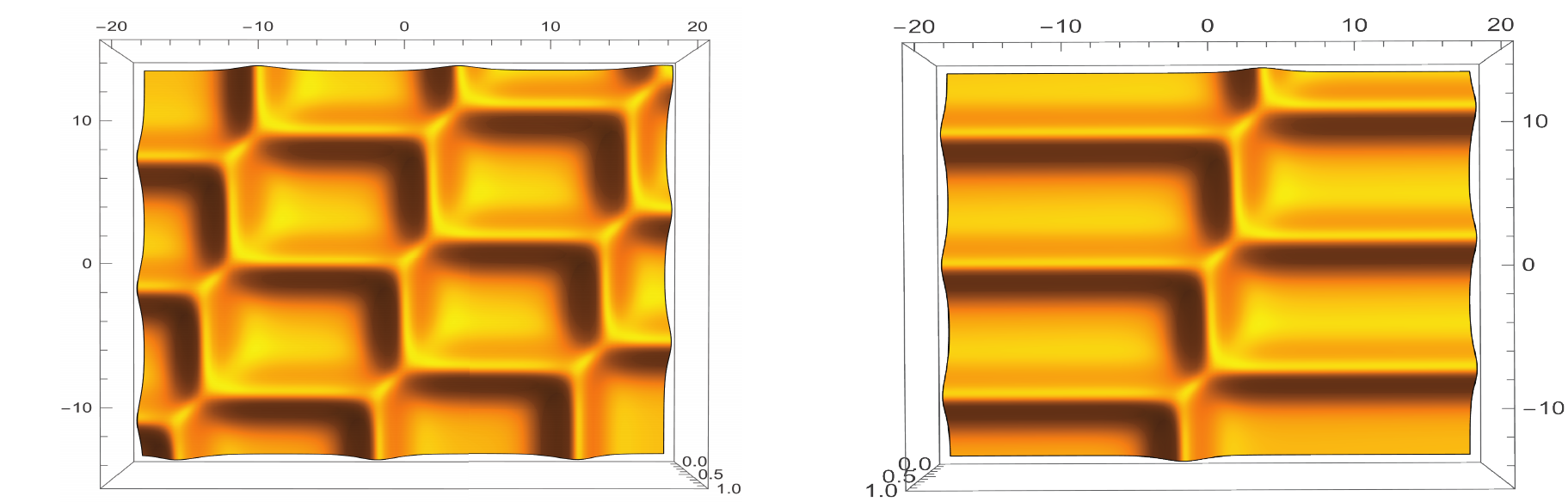}
\caption{The 3-dimensional plots of $u(\varphi;T,C)$. The left one is for $T_1=15, T_2=10$ and $\log \tilde{C}_{1,2}=-2$. The right one is for $T_1=\infty$, and $T_2=10$ and $\log \tilde{C}_{1,2}=-2$.\label{fig:CE-2}}
\end{figure}

\begin{remark}
For Figure \ref{fig:CE-2}, we used the following truncation,
\begin{align}\label{eq:g2Fig}
f[1,1](\varphi;T,C)&:=f[0,0](\varphi;C)+e^{-T_1}\left(e^{2\varphi_1}+e^{-\varphi_1}+C_{1,2}^2e^{2\varphi_1+\varphi_2}+C_{1,2}^{-1}e^{-\varphi_1+\varphi_2}\right)\\ \nonumber
&+e^{-T_2}\left(e^{2\varphi_2}+e^{-\varphi_2}+C_{1,2}^2e^{\varphi_1+2\varphi_2}+C_{1,2}^{-1}e^{\varphi_1-\varphi_2}\right)\\\nonumber
&+e^{-T_1-T_2}\left(C_{1,2}^4e^{2\varphi_1+2\varphi_2}+C_{1,2}^{-2}e^{2\varphi_1-\varphi_2}+C_{1,2}^{-2}e^{-\varphi_1+2\varphi_2}+C_{1,2}e^{-\varphi_1-\varphi_2}\right).\nonumber
\end{align}
In this truncation, we ignore the profile of $u(\varphi;T,C)$ outside the region $(-20,20)\times (-15,15)$
(recall that $T_1=15$ and $T_2=10$), since we have $|u(\varphi;T,C)-u[1,1](\varphi;T,C)|\approx e^{-\frac{T_2}{2}}=e^{-5}$ in this region. 
\end{remark}



\section{Quasi-periodic solutions to the KP equation: Application}\label{sec:KP}

The KP equation is the partial differential equation for $u=u(x,y,t)$ given by
\begin{equation}\label{eq:KP}
\partial_x(-4\partial_tu+6u\partial_xu+\partial_x^3u)+3\partial_y^2u=0,
\end{equation}
where $\partial_z^n$ are the partial derivatives such as $\partial^n_z=\frac{\partial^n}{\partial z^n}$.
The real and regular quasi-periodic solution, or finite gap solution, can be obtained in terms of the
Riemann theta function \eqref{eq:realRT} (see e.g. \cite{BBEIM:94, DN:89}). In \cite{BB:89},
the explicit formula of the solution $u(x,y,t)$ using the Schottky group is given by
\begin{equation}\label{eq:QPS}
u(x,y,t)=2\,\partial_x^2 \,\log\,f(\varphi;T,C)+2B,\quad\text{with}\quad \varphi={\bf U}^1x+{\bf U}^2y+{\bf{U}}^3t+{\bf{D}}\in \R^g,
\end{equation}
where ${\bf U}^k=(U^k_{p}: 1\le p\le g)\in\R^g$ for $k=1,2,3$ are given by
\begin{align*}
U_{p}^k&:=\sum_{\gamma\in\Gamma/\langle\gamma_{p}\rangle}\left(\gamma(\beta_p)^k-\gamma(\alpha_p)^k\right)=\beta_p^k-\alpha_p^k+\mathcal{O}(\mu).
\end{align*}
and ${\bf D}\in\R^g$ is an arbitrary $g$-dimensional vector. The constant $B$ is computed as
\[
B=\sum_{\gamma\in\Gamma\setminus \{1\}}\frac{1}{c_{\gamma}^2}=2\sum_{p=1}^g\left(\frac{\beta_p-\alpha_p}{1-\mu_{p}}\right)^2{\mu_p} +\mathcal{O}(\mu^2),
\]
where $c_\gamma$ is the element in $\gamma=\begin{pmatrix} a_{\gamma} &b_{\gamma}\\ c_{\gamma} &d_{\gamma}\end{pmatrix}\in \text{PSL}_2(\R)$.
Now it is easy to confirm that the quasi-periodic solution $u(x,y,t)$ given in \eqref{eq:QPS} leads to 
certain class of the KP soliton in the limit with all $\mu_p\to 0$ (see also \cite{K:24}).
This class of KP solitons is referred to as \emph{Hirota-type}, and in the next section we discuss
these solitons in some details in the point of view of the totally nonnegative Grassmannians (see e.g. \cite{K:17} for the details of the totally nonnegative Grassmannian with application to soliton solutions).

\subsection{KP solitons of Hirota-type}
We here provide a brief review of the {Hirota-type} $g$-solitons, which are labeled by
$g$-pairs $\{(\alpha_p,\beta_p):p=1,\ldots,g\}$ (in the Schottky data with $T_p\to\infty$ for all $p$).
We call the set of these $g$-pairs the \emph{soliton data}.
The Hirota-type solitons are then computed from \eqref{eq:TrancRTg} with ${\bf M}={\bf 0}$, that is, $f[{\bf 0}](\varphi;\tilde{C})=f(\varphi;\infty,\tilde{C})$. 
As mentioned before (see also \cite{K:17}), the regularity 
implies that the soliton data is a set of {non-crossing} pairs, which also leads
the positivity of $\tilde{C}_{p,q}$, i.e.
\[
\tilde{C}_{p,q}=\frac{(\alpha_p-\alpha_q)(\beta_p-\beta_q)}{(\alpha_p-\beta_q)(\beta_p-\alpha_q)}>0.
\]

For later use, we describe the 
soliton data using a Young diagram (this is a special case of the $\Le$-diagram introduced by
Postnikov in \cite{P:06}). First define the set of $2g$ ordered numbers $\{\kappa_1<\kappa_2<\ldots<\kappa_{2g}\}$, so that
\[
\{\alpha_1,\ldots,\alpha_g,\beta_1,\ldots,\beta_g\}=\{\kappa_1,\kappa_2,\ldots,\kappa_{2g}\}.
\]
For each pair $(\alpha_p,\beta_p)$, we have unique assignment $\alpha_p=\kappa_{i_p}$ and $\beta_p=\kappa_{j_p}$. Without loss of generality, we assume the order $\alpha_1<\alpha_2<\cdots<\alpha_g$. This implies that $\alpha_1=\kappa_1$ and $\alpha_g<\kappa_{2g}$.
Consider a square box divided into $g^2$ small boxes.
Then we take a lattice path starting from right corner of the box and ending at the bottom left corner 
with the label $\{1,\ldots, 2g\}$, so that the indices $i_p$ for $\alpha_p$ appear at the vertical paths,
and $j_p$ for $\beta_p$ appear at the horizontal paths.
Taking a horizontal path from the top right corner, and connecting it with a vertical path from the bottom
left corner, we have a Young diagram. Each pair $(\alpha_p=\kappa_{i_p},\beta_p=\kappa_{j_p})$ then gives the coordinates $(i_p,j_p)$ for a small box in the diagram, and we assign the number $p$ for those boxes.
We mark other boxes with  \raisebox{0.1cm}{\hskip0.13cm\circle{7}\hskip-0.15cm}.
For example, the  Young diagrams corresponding to the soliton data in Figure \ref{fig:genus2}  and Figure \ref{fig:g=3Ex} below are given by
\[
\young[2,2][12][1,\wb,\wb,2],\hskip1.5cm \young[2,1][12][\wb,1,2],\hskip1cm\text{and}\hskip1cm \young[3,3,1][12][\wb,1,\wb,\wb,\wb,2,3].
\]
The left diagram shows the data $\{(\kappa_1,\kappa_4),(\kappa_2,\kappa_3)\}$. The middle one shows
$\{(\kappa_1,\kappa_2),(\kappa_3,\kappa_4)\}$, and the right one for $\{(\kappa_1,\kappa_4),(\kappa_2,\kappa_3),(\kappa_5,\kappa_6)\}$.
 That is, each box represents
a single soliton in the Hirota $g$-soliton, which we call $[i_p,j_p]$-soliton (see e.g. \cite{K:17}),
or sometimes we call it $[\alpha_p,\beta_p]$-soliton.
The $[\alpha_p,\beta_p]$-soliton is (locally) expressed in the form,
\begin{align}\label{eq:1sol}
u(x,y,t)&=\frac{(\beta_p-\alpha_p)^2}{2}\sech^2\frac{1}{2}\left((\beta_p-\alpha_p)x+(\beta_p^2-\alpha_p^2)y+(\beta_p^3-\alpha_p^3)t+\varphi_p^0\right).
\end{align}
where $\varphi_p^0$ is a constant.
One should note that the $\Le$-diagram for Hirota-type soliton solution also gives the homological cycles
for the corresponding quasi-periodic solutions, that is, having with the data $(T_1,\ldots,T_g)$, we have
the Schottky data for the underlying Riemann surface. This implies that the $\Le$-diagram for the
Hirota-type soliton also describes the corresponding quasi-periodic solution associated with the Riemann surface uniformized by this Schottky group (see the examples below).

\begin{remark}
In general, the $\Le$-diagram introduced by Postnikov \cite{P:06} is defined as a decorated Young diagram
 with \raisebox{0.1cm}{\hskip0.13cm\circle{7}\hskip-0.15cm} in some boxes, which satisfies 
 the property (called {\it{$\Le$-property}}): If there is \raisebox{0.1cm}{\hskip0.13cm\circle{7}\hskip-0.1cm}, then all the boxes either to its left or above it are all \raisebox{0.1cm}{\hskip0.13cm\circle{7}\hskip-0.1cm}.
That is, there is no such \raisebox{0.1cm}{\hskip0.13cm\circle{7}\hskip-0.1cm}, which has an empty box to its left and an empty box above it. 
Postnikov then showed that each element of the totally nonnegative Grassmannian, denoted by $\text{Gr}(N,M)_{\ge0}$, can be
parametrized by the $\Le$-diagram (see also \cite{KW:14, K:17} in the connection to the KP solitons).
Note that the $\Le$-diagram for Hirota-type $g$-soliton solution corresponds to a $g$-dimensional element of $\text{Gr}(g,2g)_{\ge0}$, which is the smallest dimensional element of the \emph{irreducible} one.

It has been shown in \cite{CK:08, KW:14} that the (regular) KP solitons are known to be classified in terms of the totally nonnegative Grassmannians.
In \cite{K:24}, we showed that these KP solitons can be obtained by the $M$-theta functions
with appropriate choices of the pair $(\alpha_p,\beta_p)$. Then in \cite{IK:25}, we use
the $\Le$-diagram to connect each soliton solution with a quasi-periodic solution. However,
we found that in some cases, the quasi-periodic solutions become singular following the method developed in \cite{IK:25}.
\end{remark}


\subsection{Example 1: $g=3$ KP solution}\label{sec:KPg=3}
Here we give the results for a $g=3$ quasi-periodic solution to the KP equation. 
We consider the solution near soliton limit assuming 
large $T_p$ (recall that $e^{-T_p}=\mu_p(1+\mathcal{O}(\mu))$, Proposition \ref{prop:PeriodI}).
There are 5 different types of $g=3$ (the Catalan number $C_3$) solutions corresponding to the $\Le$-diagrams,
\[
\young[3,3,3][12][1,\wb,\wb,\wb,2,\wb,\wb,\wb,3],\qquad \young[3,3,2][12][1,\wb,\wb,\wb,\wb,2,\wb,3],\qquad  \young[3,3,1][12][\wb,1,\wb,\wb,\wb,2,3],\qquad \young[3,2,2][12][\wb,\wb,1,2,\wb,\wb,3],\qquad
\young[3,2,1][12][\wb,\wb,1,\wb,2,3].
\]

\noindent
Here we consider the solution associated with the middle diagram, whose 
 Schottky parameters $\{(\alpha_p,\beta_p,T_p):p=1,2,3\}$ are given by the table below.
\begin{table}[h]
			\centering
			\begin{tabular}{|c||c|c|c|}
				\hline
				$p$ & $1$ & $2$ &$3$   \\  
				\hline 
				$\alpha_p$   & $\kappa_{1}$ & $\kappa_{2}$ & $\kappa_{5}$   \\ 
				\hline  
				$\beta_p$ & $\kappa_{4}$ & $\kappa_{3}$  & $\kappa_{6}$   \\   
				\hline
				$T_p$ &  $T_1$ & $T_2$ &$T_3$\\
				\hline
			\end{tabular}$\qquad \text{with}\quad  \kappa_1<\kappa_2<\cdots <\kappa_6$.
		\end{table}

\noindent
That is, in the soliton limit (i.e. all $T_p \to \infty$), we have the Hirota-type soliton with $[1,4]$-, $[2,3]$-, and $[5,6]$-solitons.
Figure \ref{fig:g=3Ex} illustrates the homological cycles associated with this Schottky group.
\begin{figure}[h]
\centering
\includegraphics[width=10cm]{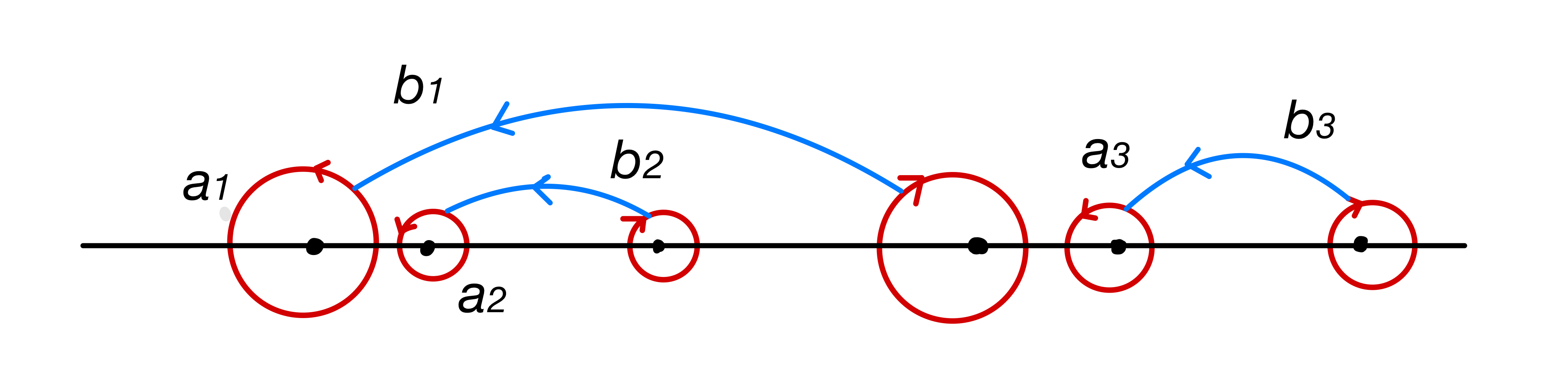}
\caption{The Riemann surface with $g=3$ uniformized by the Schottky group. The dots on
the real line indicate the soliton data in terms of the $\kappa_k$'s, see the table. \label{fig:g=3Ex}}
\end{figure}

We consider the following parameters,
\begin{align*}
&\alpha_1=\kappa_1=-0.5,\qquad \alpha_2=\kappa_2=-0.49995, \qquad \beta_2=\kappa_3=0.1,\\ &\beta_1=\kappa_4=0.5,
\qquad \alpha_3=\kappa_5=0.501,\qquad\beta_3=\kappa_6=0.9,
\end{align*}
from which, we have
\[
\tilde{C}_{1,2}=3.3335\times 10^{-6},\qquad \tilde{C}_{1,3}=286,\qquad \tilde{C}_{2,3}=1.42643.
\]
Although two parameters $\alpha_1$ and $\alpha_2$ are very close, we confirm that
using \eqref{eq:realPs},  the isometric circles $I(\gamma_1^{-1})$ and $I(\gamma_2^{-1})$ have no overlap, since we have $\alpha_1^+=-0.499996<\alpha_2^-=-0.499977$.

With the soliton data $\{\kappa_1,\ldots,\kappa_6\}$,  $[1,4]$-, $[2,3]$- and $[5,6]$-solitons
are explicitly computed by \eqref{eq:1sol}, i.e. 
\begin{align*}
& u_{[1,4]}(x,y,t)=\frac{1}{2}\sech^2\frac{1}{2}(x+0.25 t+\varphi_{1}^0),\\[0.5ex]
& u_{[2,3]}(x,y,t)=\frac{0.36}{2}\sech^2\frac{1}{2}(0.6 x-0.24 y+0.126 t+\varphi_{2}^0),\\[0.5ex]
&u_{[5,6]}(x,y,t)=\frac{0.159}{2}\sech^2\frac{1}{2}(0.4 x+0.559 y+0.603 t+\varphi_{3}^0),
\end{align*}
where $\varphi_{p}^0$ for $p=1,2,3$ are constants.
In the figures below, these solitons should be easily identified.

\begin{remark}\label{rem:resonance}
We choose $\alpha_1=\kappa_1$ be very close to $\alpha_2=\kappa_2$, so  that there is a large phase shift due to $\log\tilde{C}_{1,2}$ at the interaction between $[1,4]$-soliton and $[2,3]$-soliton. It is interesting to note that the phase shift part appears to be almost a soliton of $[3,4]$-type,
\[
u_{[3,4]}(x,y,t)=\frac{0.16}{2}\sech^2\frac{1}{2}(0.4 x+0.24 y+0.124 t+\varphi_4^0)
\]
where $\varphi_4^0$ is a constant. 
As will be discussed in Section \ref{sec:discussion}, indeed the phase shift part 
becomes this soliton (this may be explained as a result of \emph{degeneration} of the $\Le$-diagram,
which will be discussed in a future communication).
Also a large phase shift can be seen between $[1,4]$-soliton and $[5,6]$-soliton (see the figures below) due to a large value of $\tilde{C}_{1,3}$. This is because $\beta_1=\kappa_4$ is close to $\alpha_3=\kappa_5$.
Then we may observe the soliton of $[1,6]$-type as a phase shift,
\[
u_{[1,6]}(x,y,t)=\frac{1.96}{2}\sech^2\frac{1}{2}(1.4 x+0.56 y+0.854 t+\varphi_5^0),
\]
where $\varphi^0_5$ is a constant.
\end{remark}

The $g=3$ quasi-periodic solution of the KP equation is given by \eqref{eq:QPS}. Since the constant term
$B$ in the formula is of order $\mathcal{O}(\mu)$, one can ignore $B$ in our computation with $T_p\approx 25$.
In Figures \ref{fig:Ex30}, \ref{fig:Ex21}, \ref{fig:Ex21c}, \ref{fig:Ex12} and \ref{fig:Ex03} below, we show the profiles of $u(x,y,t)=2\partial_x^2\log f(\varphi;T,C)$ for the region $\mathcal{D}=(-30,50)\times (-80,60)$ in the $xy$-plane at $t=-20,-10,0$. 
In each figure, the Schottky parameters $(T_1,T_2,T_3)$ are chosen as follows: 
\begin{itemize}
\item[(a)] $(25,20,25)$ for Figure \ref{fig:Ex30} showing $g=3$ quasi-periodic solution,
\item[(b)] $(25,20,100)$ for Figure \ref{fig:Ex21} showing $[5,6]$-soliton solution on $g=2$ periodic background,
\item[(c)] $(100,20,25)$ for Figure \ref{fig:Ex21c} showing $[1,4]$-soliton solution on $g=2$ periodic background,
\item[(d)] $(25,100,100)$ for Figure \ref{fig:Ex12} showing $[2,3]$-, and $[5,6]$-solitons on elliptic background,
\item[(e)] $(100,100,100)$ for Figure \ref{fig:Ex03} showing Hirota-type 3-soliton with $[1,4]$-. $[2,3]$-, and $[5,6]$-type solitons. The phase shift parts may be identified as $[3,4]$- and $[1,6]$-solitons.
\end{itemize}
As was mentioned, in this size of computational window $\mathcal{D}=(-30,50)\times (-80,60)$, we use $f[2,2,2](\varphi;T,C)$
to compute $u(x,y,t)$,
\[
u(x,y,t)=2\partial_x^2\log f[2,2,2](\varphi;T,C) +\mathcal{O}(e^{-\frac{T_0}{2}}) \quad\text{for}\quad (x,y)\in \mathcal{D}\subset\R^2,
\]
where $f[2,2,2](\varphi;T,C)$ is the ${\bf M}=(2,2,2)$-truncation of the theta function \eqref{eq:realRT}.
Note that within the region $\mathcal{D}$, the Schottky parameter $T_p=100$ can be considered as $T_p=\infty$.
\begin{figure}[h]
\centering
\includegraphics[width=12cm] {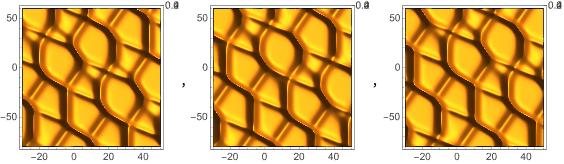}
\caption{The profiles of $u(x,y,t)$ at $t=-20,-10, 0$. The parameters $(T_1,T_2,T_3)=(25,20,25)$.
 This is $g=3$ quasi-periodic solution.  \label{fig:Ex30}}
\end{figure}
\begin{figure}[h]
\centering
\includegraphics[width=12cm] {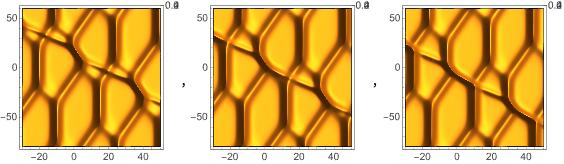}
\caption{The profile of $u(x,y,t)$ at $t=-20,-10,0$ with $(T_1,T_2,T_3)=(25,20,100)$.
The solution shows $[5,6]$-soliton on $g=2$ quasi-periodic background.
\label{fig:Ex21}}
\end{figure}
\begin{figure}[h]
\centering
\includegraphics[width=12cm] {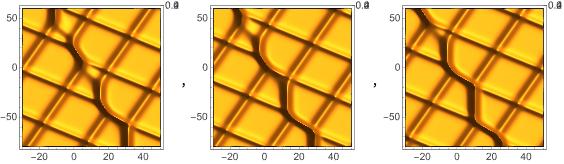}
\caption{The profile of $u(x,y,t)$ at $t=-20,-10,0$ with $(T_1,T_2,T_3)=(100,20,25)$.
The solution shows $[1,4]$-soliton on $g=2$ quasi-periodic background.
\label{fig:Ex21c}}
\end{figure}
\begin{figure}[h]
\centering
\includegraphics[width=12cm] {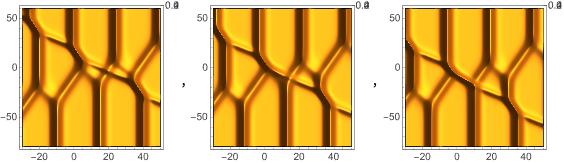}
\caption{ The profile of $u(x,y,t)$ at $t=-20,-10,0$ with  $(T_1,T_2,T_3)=(25,100,100)$.
The solution is 2-soliton with $[2,3]$- and $[5,6]$-solitons on elliptic background. \label{fig:Ex12}}
\end{figure}
\begin{figure}[h]
\centering
\includegraphics[width=12cm] {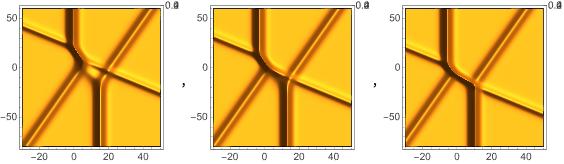}
\caption{ The profile of $u(x,y,t)$ at $t=-20,-10,0$ with $(T_1,T_2,T_3)=(100,100,100)$. 
The solution is 3-soliton of Hirota-type with $[1,4]$-, $[2,3]$- and $[5,6]$-type solitons. 
At the interaction regions, $[3,4]$- and $[1,6]$-solitons are observed.\label{fig:Ex03}}
\end{figure}

\subsection{Example 2: $g=2$ KdV solution}
The KdV equation can be considered as the 2 reduction of the KP equation, which means  that the solution
$u(x,y,t)$ does not depend on the $y$-variable (see e.g. \cite{HKL:25}). 
This can be imposed by taking $\alpha_p=-\beta_p$ and $\gamma_p(z)=-\gamma_p^{-1}(-z)$, 
which gives ${\bf U}^2={\bf 0}$ in \eqref{eq:QPS}. Then
we take the Schottky data as in the table below with 
\[
\alpha_1=\kappa_1=-1,\quad \alpha_2=\kappa_2=-0.7,\quad \beta_2=\kappa_3=-\kappa_2=0.7,\quad\beta_1=\kappa_4=-\kappa_1=1,
\]
from which we have
\[
\tilde{C}_{1,2}=\left(\frac{\kappa_1-\kappa_2}{\kappa_1+\kappa_2}\right)^2=0.0311<1.
\]
\begin{table}[h]
			\centering
			\begin{tabular}{|c||c|c|}
				\hline
				$p$ & $1$ & $2$    \\  
				\hline 
				$\alpha_p$   & $\kappa_{1}$ & $\kappa_{2}$   \\ 
				\hline  
				$\beta_p$ & $\kappa_{4}$ & $\kappa_{3}$     \\   
				\hline
				$T_p$ &  $T_1$ & $T_2$ \\
				\hline
			\end{tabular}$\qquad \text{with}\quad  \kappa_1<\kappa_2<0<\kappa_3=-\kappa_2<\kappa_4=-\kappa_1$.
		\end{table}

\noindent
Taking large $T_1$ and $T_2$, we consider the quasi-periodic solution near soliton limit, 
and we ignore the constant term $B$ in the solution formula \eqref{eq:QPS}.
Figure \ref{fig:KdV2} shows the 3 dimensional plot of the solution $u(x,t)$ with $T_1=T_2=20$ for
$-15\le x\le 30$ and $-15\le t\le 9$ (left), and the profiles of $u(x,t)$ at $t=3 s$ with $-5\le s\le 3$ (right).
\begin{figure}[h]
\centering
\includegraphics[width=12.5cm] {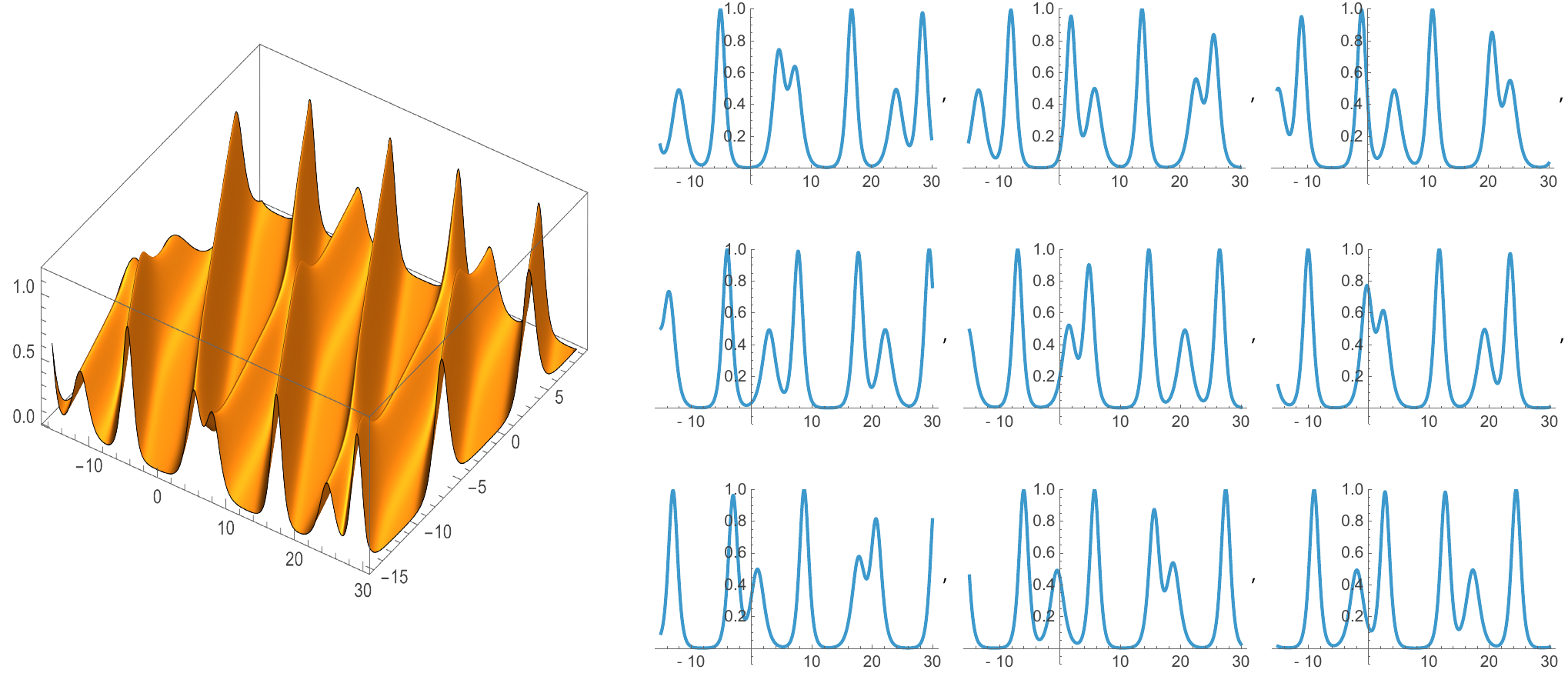}
\caption{The $g=2$ quasi-periodic solution of the KdV equation. The left panel shows
the 3D plot of the solution $u(x,t)$ on $(-15,30)\times (-15,9)$ in the $xt$-plane. The right panels show the profiles
of the solution $u(x,t)$ at $t=3s$ with $-5\le s\le 3$ from the left to right and top to bottom.\label{fig:KdV2}}
\end{figure}
Figure \ref{fig:KdV-Tess1} shows the tessellation of $\R^2$ with a parallelogram tiling, and we have
a 2-torus,
\[
\R^2/\mathcal{T}~\cong~S^1\times S^1,
\]
where the group of translation is $\mathcal{T}=\langle\mathcal{T}_1,\mathcal{T}_2\rangle$ with
the generators
\begin{align}
\mathcal{T}_1~:~\left\{\begin{array}{lll}
\varphi_1~\longrightarrow~\varphi_1+T_1,\\[1.0ex]
\varphi_2~\longrightarrow~\varphi_2-\log C_{1,2},
\end{array}\right.\qquad\mathcal{T}_2~:~\left\{\begin{array}{lll}
\varphi_1~\longrightarrow~\varphi_1-\log C_{1,2},\\[1.0ex]
\varphi_2~\longrightarrow~\varphi_2+T_2.
\end{array}\right.
\end{align}
The phase variables $\varphi_p$ are 
\[
\varphi_p=2\beta_px+2\beta_p^3t,\qquad p=1,2.
\]
Then one can find the coordinates of the point $B$ in Figure \ref{fig:KdV-Tess1}, that is, the point $B$
is the intersection point of the two lines,
$\varphi_1=T_1$ and $ \varphi_2=-\log C_{1,2}$,
and i the coordinates of $B$ is given by,
\[
x=\frac{\beta_2^3T_1+\beta_1^3\log C_{1,2}}{2\beta_1\beta_2(\beta_2^2-\beta_1^2)}=-4.75,
\quad\text{and}\quad
t=\frac{\beta_2T_1+\beta_1\log C_{1,2}}{2\beta_1\beta_2(\beta_1^2-\beta_2^2)}=14.75
\]
Since the points $(0,0)$ and $B$ are the same in 2-torus, the $t$-period is given by $14.75$ with the $x$-shift $-4.75$. This can be checked in Figure \ref{fig:KdV2}. For example, the profile at $t=-15$ is almost the same as that at $t=0$ with the shift in the $x$-direction.
\begin{figure}[h]
\centering
\includegraphics[width=12cm] {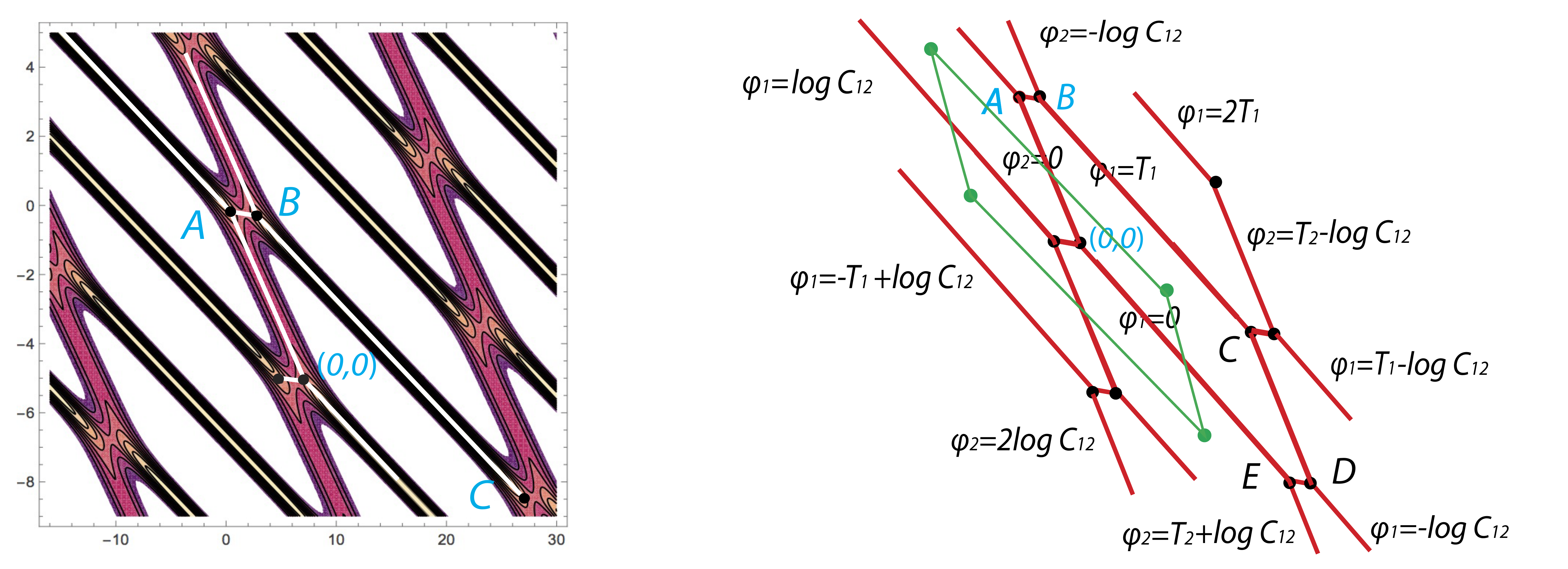}
\caption{The left panel shows contour plot of the solution $u(x,t)$ of the KdV equation, and the right one expresses the tessellation of $\R^2$ with the parallelogram (green), or the hexagon (red)
obtained from the contour plot.\label{fig:KdV-Tess1}}
\end{figure}
\begin{remark}
The quasi-periodic solution considered here is associated with  $g=2$ hyper-elliptic Riemann surface
given by the algebraic curve in $(\lambda,\mu)\in \C^2$,
\[
\mu^2=\lambda(\lambda-E_1)(\lambda-E_2)(\lambda-E_3)(\lambda-E_4),
\]
where $E_i$'s are real and positive with the order $E_4<E_3<E_2<E_1$. 
In \cite{BBEIM:94}, these parameters are given by
\begin{equation}\label{eq:KdVSchottky}
\left\{\begin{array}{lll}
E_{2p-1}=(\beta_p^+)^2+\sum_{\gamma\in \Gamma, \gamma\ne Id}\left(\gamma(\beta_p^+)^2-\gamma(0)^2\right),\\[3.0ex]
E_{2p}=(\beta_p^-)^2+\sum_{\gamma\in \Gamma, \gamma\ne Id}\left(\gamma(\beta_p^-)^2-\gamma(0)^2\right),
\end{array}\right.\quad\text{for}~p=1,2,
\end{equation}
where $\beta_p^{\pm}$ are the real points on the isometric circle $I(\gamma_p)$ given in \eqref{eq:realPs}, and they are
\[
\beta_p^+=\frac{1+\sqrt{\mu_p}}{1-\sqrt{\mu_p}}\beta_p,\qquad \beta_p^-=\frac{1-\sqrt{\mu_p}}{1+\sqrt{\mu_p}}\beta_p.
\]
where we have used the KdV condition $\alpha_p=-\beta_p<0$.
This leads to
\begin{align*}
\beta_p=-\alpha_p=\sqrt{\beta_p^+\beta_p^-}\quad \text{and}\quad
\sqrt{\mu_p}=\frac{\sqrt{\beta_p^+}-\sqrt{\beta_p^-}}{\sqrt{\beta_p^+}+\sqrt{\beta_p^-}}, \qquad\text{for}\quad p=1,2.
\end{align*}
Note here that the pinching $a_p$-cycles, i.e. $\mu_p\to 0$,  corresponds to the degeneration of the branch points
$E_{2p-1}=E_{2p}=\kappa_{p}^2$. The relation \eqref{eq:KdVSchottky} holds for the surface with arbitrary genus,
\[
\mu^2=\lambda\prod_{p=1}^g(\lambda-E_{2p-1})(\lambda-E_{2p})\quad \text{with}\quad E_i<E_j\quad\text{for}\quad i>j.
\]

\end{remark}

\subsection{Example 3: $g=3$ KdV solution}
Here we consider a  $g=3$ quasi-periodic solution and solitons on quasi-periodic background of $g<3$.
 We take the parameters shown in the table.
\begin{table}[H]
			\centering
			\begin{tabular}{|c||c|c|c|}
				\hline
				$p$ & $1$ & $2$  &$3$  \\  
				\hline 
				$\alpha_p$   & $\kappa_{1}$ & $\kappa_{2}$  & $\kappa_3$  \\ 
				\hline  
				$\beta_p$ & $\kappa_{6}$ & $\kappa_{5}$    & $\kappa_4$ \\   
				\hline
				$T_p$ &  $T_1$ & $T_2$  & $T_3$\\
				\hline
			\end{tabular}$\qquad \text{with}\quad  \kappa_1<\kappa_2<\kappa_3 <0<\kappa_4=-\kappa_3<\kappa_5=-\kappa_2<\kappa_6=-\kappa_1$.
		\end{table}

\noindent
The hyper-elliptic curve of $g=3$ associated with this case is given by
\begin{equation}\label{eq:g=3KdVC}
\mu^2=\lambda\prod_{k=1}^6(\lambda-E_k),
\end{equation}
where the branch points are given by \eqref{eq:KdVSchottky} in terms of the Schottky parameters.
We are interested in the case near soliton limit, and assume all $T_p\ge 20$,
which means $\mu\approx e^{-20}$.
Within the error of $\mu$, the Schottky parameters are well approximated by the soliton parameters,
and ${C}_{p,q}=\tilde{C}_{p,q}(1+\mathcal{O}(\mu))$. For the values $T_p\ge 20$, one can use 
the formula \eqref{eq:TrancRTg} with ${\bf M}=(M_1,M_2,M_3)=(2,2,2)$. 
 
Figure \ref{fig:KdVg=3} illustrates the contour plots of the solution $u(x,t)$ over $(-20,25)\times (-26,30)$
in the $xt$-plane. We take the $\kappa$-parameters satisfying $\kappa_i=-\kappa_{7-i}$ for $i=1,2,3$ with
\[
\kappa_1=-1~<~\kappa_2=-0.8~<~\kappa_3=-0.5,
\]
which give
\[
C_{1,2}=0.012,\quad C_{1,3}=0.111,\quad C_{2,3}=0.0532.
\]
In the figures, the Schottky parameters $(T_1,T_2,T_3)$ for the panels (a), (b), (c) and (d) are chosen as follows.
\begin{itemize}
\item[(a)] $(20,20,20)$. This is $g=3$ quasi-periodic solution for the hyper-elliptic curve 
of \eqref{eq:g=3KdVC}.
\item[(b)] $(200,20,20)$.  This is 1-soliton of $[1,6]$-type on $g=2$ quasi-periodic background for the
degenerate curve with $E_1=E_2$, i.e. 
\[\mu^2=\lambda(\lambda-E_1)^2\prod_{k=3}^6(\lambda-E_k).
\]
\item[(c)] $(200,20,150)$. This is 2-soliton with $[1,6]$- and $[3,4]$-types on elliptic background for
the degenerate curve with $E_1=E_2$ and $E_5=E_6$, i.e. 
\[\mu^2=\lambda(\lambda-E_1)^2(\lambda-E_3)(\lambda-E_4)(\lambda-E_5)^2.
\]
\item[(d)] $(200,150,150)$. This is a Hirota 3-soliton with $[1,6]$-, $[2,5]$-and $[3,4]$-types for the singular curve
\[
\mu^2=\lambda(\lambda-E_1)^2(\lambda-E_3)^2(\lambda-E_5)^2.
\]
\end{itemize}

\begin{figure}[h]
\centering
\subfigure[$:E_i\ne E_j$]{
\includegraphics[width=3.2cm] {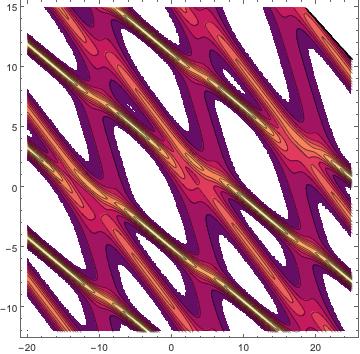}}
\subfigure[$:E_1=E_2$]{
\includegraphics[width=3.2cm]{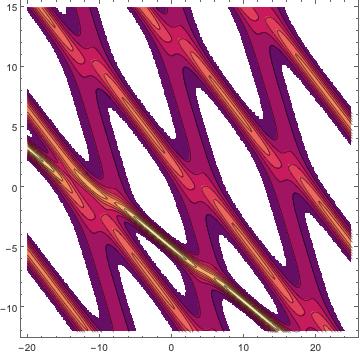}}
\subfigure[:$E_1=E_2, E_5=E_6$]{
\includegraphics[width=3.2cm]{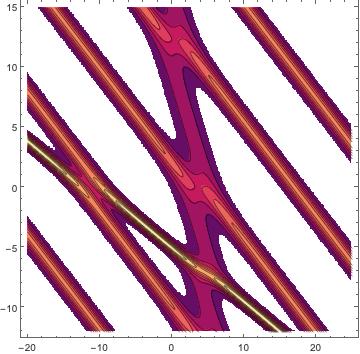}}
\subfigure[All pinched]{
\includegraphics[width=3.2cm]{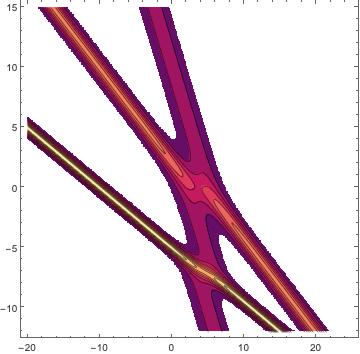}}
\caption{Contour plots of the $g=3$ quasi-periodic solution $u(x,t)$ 
and its degenerations (solitons on quasi-periodic background) of the KdV equation. 
\label{fig:KdVg=3}}
\end{figure}

\noindent

\section{Discussion}\label{sec:discussion}
In this paper, we studied the real Riemann theta function near soliton limit, and 
as the applications, we considered regular
quasi-periodic solutions of the KP equation and the KdV equation (as the simplest
reduction of the KP equation). We then found that the solution associated with the Riemann surface of genus $g$ induces a tessellation of $\R^g$, whose tile is divided into $2^g$ sections by the soliton phases.
The real and regular soliton solutions of the KP equation, called the KP solitons, have been classified
in terms of the totally nonnegative Grassmannians, denoted by $\text{Gr}(N,M)_{\ge0}$ with $N<M$,
the set of $N\times M$ full rank matrices having with all nonnegative minors (see e.g. \cite{K:17} for a review of the KP solitons).

One of the main motivations of the present study is to give an explicit description
of the Grassmannian structure in the Riemann surface. In other words,
we wish to find (locally) the KP solitons in the quasi-periodic solutions in the soliton limit.
For example, consider the case with $g=2$ quasi-periodic solutions of the KP equation associated with
the $\Le$-diagram \young[2,2][6][{\tiny 1},\wbs,\wbs,{\tiny 2}].  The coefficient $C_{1,2}$, related to
the off-diagonal element of the period matrix near the soliton limit, is given by
\[
C_{1,2}=\frac{(\alpha_1-\alpha_2)(\beta_1-\beta_2)}{(\alpha_1-\beta_2)(\beta_1-\alpha_2)}\left(1+\mathcal{O}(\mu)\right).
\]
We choose the parameters as in Example \ref{sec:KPg=3}, i.e.
\[
\alpha_1=-0.5,\quad \alpha_2=-0.49995, \quad \beta_2=0.1,\quad \beta_1=0.5,
\]
we have ${C}_{1,2}=3.33335\times 10^{-6}(1+\mathcal{O}(\mu))$. 
We take $T_1=25$ and $T_2=20$, which gives the error $\mu=e^{-20}=2.061\times 10^{-9}$.
The left panel in Figure \ref{fig:KP2} shows the solution $u(x,y,t)$, which is stationary in time.
The middle one shows the corresponding 2-soliton solution. The right one is the resonant Y-soliton solution
generated by taking the limit $\alpha_1\to\alpha_2$, that is, $C_{1,2}\to 0$ (leading to a large phase shift).
As mentioned in Remark \ref{rem:Off}, the singularity corresponding to Y-soliton is an ordinary \emph{triple} singularity due to the limit $\alpha_1\to\alpha_2$,
\begin{figure}[h]
\centering
\subfigure[$g=2$ quasi-periodic solution]{
\includegraphics[width=3.2cm] {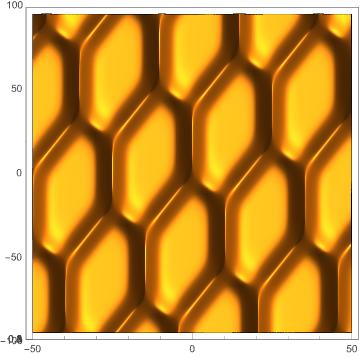}}
\hskip0.5cm
\subfigure[2-soliton solution]{
\includegraphics[width=3.2cm]{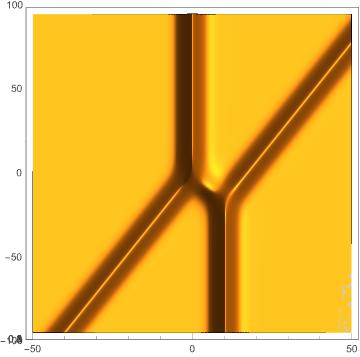}}
\hskip0.5cm
\subfigure[Y-soliton solution]{
\includegraphics[width=3.2cm]{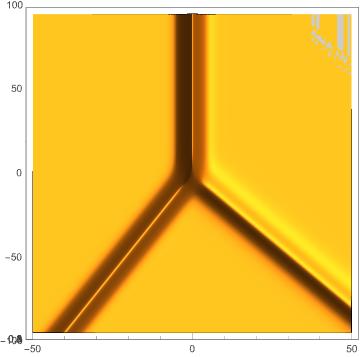}}
\caption{The profiles of the solutions $u(x,y,t)$ of the KP equation. The panel (a) shows $g=2$ quasi-periodic solution.
(b) shows the Hirota-type 2-soliton with $[1,4]$- and $[2,3]$-type solitons obtained in the soliton limit $T_p\to \infty$ for both $p=1,2$, and (c) shows the resonant Y-soliton obtained by further limit  $\alpha_1\to\alpha_2$ giving $C_{1,2}\to 0$. The phase shift part becomes $[3,4]$-soliton (see Remark \ref{rem:resonance}).
\label{fig:KP2}}
\end{figure}
and the $\Le$-diagram is modified to \young[2][6][,],
which is the diagram for Gr$(1,3)_{>0}$. We call this limit a \emph{degeneration} of the $\Le$-diagram.
The Y-soliton represents a \emph{resonant} interaction of two line solitons, $[1,4]$-, $[2,3]$-solitons with
$\alpha_1=\alpha_2$. The resonant interaction generates a third soliton of $[3,4]$-type, see Remark \ref{rem:resonance} (see e.g. \cite{K:18} for the soliton resonant interactions).
One should note that one can also consider the degeneration with the limit $\beta_1\to\beta_2$.
In this case, we have the degeneration from \young[2,2][6][,\wbs,\wbs,] to \young[1,1][6][,].
The degenerated diagram \young[1,1][6][,] corresponds to Gr$(2,3)_{>0}$, and
the soliton solution for this diagram gives another resonant Y-soliton.
We will discuss the details in a future communication.
Here we just like to mention that in Example \ref{sec:KPg=3} of $g=3$ KP quasi-periodic solution, 
one can observe locally the soliton solutions of Gr$(k,4)_{\ge0}$ for $k=1,2,3$.
The degeneration of the $\Le$-diagram may be described by
\[
\young[3,3,1][10][\wb,,\wb,\wb,\wb,,]\quad\longrightarrow\quad \young[3][10][,,],\qquad \young[2,1][10][,,],\qquad \young[2,2][10][\wb,,,,],\qquad \young[2,2][10][,\wb,,], \qquad \young[1,1,1][10][,,].
\]
In terms of the Grassmannians, the cell in Gr$(3,6)_{\ge0}$ is degenerated to either the cells in Gr$(1,4)_{>0}$ or Gr$(2,4)_{\ge0}$ or Gr$(3,4)_{>0}$ with dimension $g=3$. The KP solitons corresponding to those cells can be observed in Figure \ref{fig:Ex30} (see also \cite{AG:18} for the Grassmannian structure in the theta function of genus 4). It is interesting to note that the KP soliton associated with \young[2,2][6][,,\wbs,] is missing in this example. 

The fact that the quasi-periodic solution includes
several solitons may indicate that the theta function is rather related to a flag variety, which consists of
the product of Grassmannians (the diagonal embedding), e.g. for the flag variety $SL_4(\R)/B$ with the Borel subgroup $B$, we have
\[
SL_4(\R)/B~\hookrightarrow~\text{Gr}(1,4)\times \text{Gr}(2,4)\times \text{Gr}(3,4).
\]
In particular, we note that the solutions expressed by the theta function is naturally leads to these of the Toda lattices with periodic or Dirichlet-type boundary conditions
(see e.g. \cite{KW:15, Ab:17}).
We will discuss the details in a future communication, 
where we will give an \emph{explicit} limiting algorithm to find KP solitons from a quasi-periodic solution.

\bigskip
\noindent
{\bf Acknowledgements.}
I would like to thank Atsushi Nakayashiki for valuable discussions and sending me his presentation pdf
for the workshop at Kobe University on March 17-19, 2026. I would also like to thank Sasha Mikhailov 
for useful discussion on a higher genus analogue of the elliptic case as discussed in Introduction.
I am grateful to Shilong Huang and Shangshuai Li for critical reading of the manuscript. The discussions with
 Sasha Veselov, Antonio Moro, Manuel Manas and Sarby Chakravarty are also appreciated.

\raggedright
 \bigskip
 \noindent
 {\bf Declarations.}
 \begin{itemize}
\item Data sharing not applicable to this article as no datasets were generated or analyzed during the current study.\\
 \item The author declares no conflicts of interest associated with this manuscript.
 \end{itemize}



\end{document}